%% file: samplepaper.tex
\documentclass[runningheads]{llncs}
\usepackage{graphicx}
\usepackage{settings}

\begin{document}
\title{On the Simulation Power of Surface Chemical Reaction Networks\thanks{This work is supported by NSTC(Taiwan) grant number 110-2223-E-002-006-MY3.}}
%
%
\author{Yi-Xuan Lee\inst{1} \and
Ho-Lin Chen\inst{1}\thanks{corresponding author}\\ 
%
\institute{National Taiwan University}
\email{\{r12921065,holinchen\}@ntu.edu.tw}}
\authorrunning{Y.-X. Lee and H.-L. Chen}

\maketitle              
\begin{abstract}
The Chemical Reaction Network (CRN) is a well-studied model that describes the interaction of molecules in well-mixed solutions. In 2014, Qian and Winfree~\cite{QW14} proposed the abstract surface chemical reaction network model (sCRN), which takes advantage of spatial separation by placing molecules on a structured surface, limiting the interaction between molecules. In this model, molecules can only react with their immediate neighbors. Many follow-up works study the computational and pattern-construction power of sCRNs.

In this work, our goal is to describe the power of sCRN by relating the model to other well-studied models in distributed computation. Our main result is to show that, given the same initial configuration, sCRN, affinity-strengthening tile automata, cellular automata, and amoebot can all simulate each other (up to unavoidable rotation and reflection of the pattern). One of our techniques is coloring on-the-fly, which allows all molecules in sCRN to have a global orientation.

\end{abstract}

\input{intro_v1}
\input{models_v10}
\input{simdef_v6}

\input{s-sCRN_to_s-d-sCRN_v9}

\input{sCRN-aTAM_v8}
\input{sCRN-TA_v8}
\input{sCRN-asyncCA_v6}
\input{sCRN-Amoe_v7}

\bibliographystyle{splncs04}
\bibliography{ref_v1}



%
%
%

\end{document}

%% file: intro_v1.tex
\section{Introduction}
The Chemical Reaction Network (CRN) is a well-studied model that describes the interaction of molecules in well-mixed solutions. The computational power of CRNs has been studied in various settings, including rate-dependent~\cite{SCW08,CSW09,FGB17} and rate-independent models~\cite{cdd14:stochastic,cdsw23:mass_action}. In this study, rather than considering molecules in a well-mixed solution, we study the scenarios in which molecules have designated locations and can only react with other molecules in their vicinity.


In 2014, Qian and Winfree~\cite{QW14} proposed the abstract surface chemical reaction network model (sCRN), which takes advantage of spatial separation by placing molecules on a structured surface, limiting the interaction between molecules. In this model, molecules can only react with their immediate neighbors. Qian and Winfree~\cite{QW14} also proposed an implementation of an arbitrary programmable sCRN using DNA strand displacement. Furthermore, they demonstrate the computational power of sCRN by building a continuous active logic circuit and cellular automata in a parallel and scalable way.

A follow-up work by Clamons\etal~\cite{CQW20} discuss the computational and pattern-creation power of sCRNs. They extend the idea of~\cite{QW14} to achieve local synchronicity on a $2$-dimensional square lattice, with an initial pattern providing local orientation. This requires tethering the molecules precisely on the surface. On the other hand, they ask whether we can achieve complex spatial arrangements from relatively simple initial conditions, which rely on the ability to control molecular interactions precisely. They also emphasize the trade-off between these two concepts. Some of the open questions are solved in~\cite{AB23}, which focused on the complexity of deciding the reachability from a given configuration to another given configuration. Brailovskaya et al.~\cite{BGY19} showed that feed-forward circuits can be constructed in sCRN using only swap reactions. In this work, our main goal is to study the computational power of sCRN given a uniform surface with no local orientation.

We seek to describe the power of sCRN by relating the model to other well-studied models in distributed computation, such as abstract tile assembly model (aTAM)~\cite{Win98,RW00,DL12,ALQ01}, cellular automata~\cite{JVN66,BNR20,FG08}, and tile automata~\cite{CL18,CGS21,RDS23}. These models have been shown to be Turing universal~\cite{Win98,DFM12,CGS21} and can perform a large variety of tasks, including computation and pattern formation (e.g.,\cite{Pa12,BNR20,RDS23}). We further compare sCRN to the amoebot model, where programmable matter has the ability to move. We use the concept of "simulation" based on the definition given in~\cite{HP23} with some necessary modifications. Intuitively, if two models can simulate each other, they can perform the same tasks in the same way (e.g., create the same pattern in the same ordering).

\subsection{Our Results}
In this work, our main result is to show that, given the same initial configuration, sCRN, affinity-strengthening tile automata, cellular automata, and amoebot can all simulate each other (up to unavoidable rotation and reflection of the pattern). The results are listed in Table~\ref{tab:result}, where $\Mcal_\Scal,\Mcal_\Tcal$ represent different models. We ask whether $\Mcal_\Scal$ can simulate $\Mcal_\Tcal$ (denoted by $\Mcal_\Scal\triangleright\Mcal_\Tcal$), and vice versa (denoted by $\Mcal_\Tcal\triangleright\Mcal_\Scal$). 

The main challenge for this simulation is that sCRN does not have predefined directions. Unlike the tile automata model (and all other models simulated), which has a given direction, the molecules in sCRN do not have the ability to differentiate their neighbors in different directions. We solve this problem by providing a coloring to the surface such that each molecule has four neighbors with different colors. Furthermore, in the simulation of $\Mcal_\Tcal$ by $\Mcal_\Scal$, we want to make the behavior of $\Mcal_\Scal$ to be as close to $\Mcal_\Tcal$ as possible. When $\Mcal_\Tcal$ enters a terminal configuration, $\Mcal_\Scal$ should also enter a terminal configuration without using too much extra space. Therefore, instead of coloring the whole surface, we carefully perform the coloring on the fly. A molecule is colored only when molecules in its vicinity are about to participate in a reaction. This coloring technique can be used in the simulation of all different models, so we describe this coloring technique separately in Section~\ref{sec:sCRN-dsCRN}.


In Section~\ref{sec:sCRN-dsCRN}, we show that unit-seeded sCRN (s-sCRN) can simulate unit-seeded, directed sCRN (s-d-sCRN) up to rotation and reflection, while simulating s-sCRN with s-d-sCRN is trivial. In Section~\ref{sec:aTAM}, we show that s-d-sCRN can simulate aTAM.
It is easy to show that aTAM cannot simulate sCRN since the reactions are not reversible.
In Section~\ref{sec:TA}, we show that s-d-sCRN and unit-seeded tile automata with affinity-strengthening rules (s-as-TA) can simulate each other. In Section~\ref{sec:CA}, we show that d-sCRN and asynchronous cellular automata (async-CA) with non-deterministic local function can simulate each other. We also show that clockwise sCRN (c-sCRN) and amoebot can simulate each other in section~\ref{sec:amoe}. Notice that, when we simulate s-d-sCRN with s-sCRN, we are essentially giving a random global orientation to the surface. When simulating unit-seeded models, the orientation of the seed can be included in the coloring process and thus the terminal configuration of the sCRN is always the same as the system simulated, up to rotation and reflection. On the other hand, cellular automata and amoebots have complicated initial patterns. Therefore, simulating cellular automata and amoebots by s-sCRN (undirected) will inevitably rotate or flip the initial configuration before the simulation starts, unless the direction used in the cellular automata and amoebots is also encoded in the initial configuration of s-sCRN. 

\begin{table}
\caption{Cross-Model Simulation Results.}\label{tab:result}
\begin{tabular}{|m{6em}|m{6em}|m{8em}|m{6em}|m{6em}|}
\hline
$\Mcal_\Scal$ & $\Mcal_\Tcal$ & $\Mcal_\Scal\triangleright\Mcal_\Tcal$ & $\Mcal_\Tcal\triangleright\Mcal_\Scal$  & REF \\
\hline
\hline 
s-sCRN & s-d-sCRN & \checkmark \ (up to rotation and reflection) & \checkmark  &  Thm~\ref{thm:sCRN-dsCRN}\\
\hline
s-d-sCRN & aTAM & \checkmark & $\times$ & Thm~\ref{thm:dsCRN-aTAM}\\
\hline
s-d-sCRN& s-as-TA & \checkmark & \checkmark & Thm~\ref{thm:dsCRN-TA}\\
\hline
d-sCRN & async-CA & \checkmark & \checkmark & Thm~\ref{thm:dsCRN-CA},~\ref{thm:CA-dsCRN}\\
\hline
c-sCRN & amoebot & \checkmark & \checkmark & Thm~\ref{thm:csCRN-amoe},~\ref{thm:amoe_csCRN}\\
\hline
\end{tabular}
\end{table}


%% file: models_v10.tex
\section{Models}
\subsection{Surface chemical reaction network}\label{ssec:sCRN}

The model of \textbf{surface chemical reaction network (sCRN)} was originally proposed by Qian and Winfree~\cite{QW14}. To align the underlying structures of all models and to describe the simulation problem, which is highly relative to the transformation of patterns in each system, we use a slightly different definition given in~\cite{AB23}. In particular, they define the configuration and reachability in sCRN. To make the simulations simpler, we also propose in Section~\ref{sssec:dsCRN} the directed sCRN as a variation of sCRN. We now give a brief description.

A surface chemical reaction network $(Q,R)$ consists of an underlying \emph{surface} $L$, a finite set of species $Q$, and a set of \emph{reactions} $R$. In general, the surface $L$ is an arbitrary planar graph, but we restrict it to the square lattice (viewed as $\Zbb^2$) and the triangular lattice with nearest neighbor connectivity in this paper. The \emph{cells} are vertices of $L$, and each cell $u$ is associated with a species $\sigma(u)$ in $Q$. For the purpose of defining the simulation between computational models in Section~\ref{sec:simdef}, we sometimes call the species "states" of a cell, and $Q$ the set of "states". A \emph{configuration} is a mapping from each cell to a state in $Q$. Every reaction $r\in R$ has one of the following two forms: ($A,B,C,D\in Q$)

\begin{itemize}
\item \emph{Unimolecular reaction} $A\rw B$. A cell in state $A$ can change its state to $B$ itself. Or,
\item \emph{Bimolecular reaction} $A+B\rw C+D$ where $A,B,C,D$ not necessarily distinct, meaning that when two species $A$ and $B$ are adjacent, their states could be replaced with $C$ and $D$ respectively and simultaneously. Note that the orientation of $A$ and $B$ does not matter, but the order cannot be changed (state $A$ must turn into $C$, and $B$ must turn into $D$).
\end{itemize}

In this paper, we consider the sCRN $\Gamma=(Q,S,R)$ with a specified initial configuration $S$. We further define the \textbf{unit-seeded sCRN (s-sCRN)} where the initial configuration $S$ maps every cell to a \emph{blank state} $\Ocal\in Q$ except for a special cell mapped to a \emph{seed state} $s$. In s-sCRN we don't allow any reactions s.t. $\Ocal$ is the only reactant when considering such unit-seed system. i.e. $\Ocal\rw A$ or $\Ocal+\Ocal\rw A+B$ is illegal here.

\subsubsection{Surface chemical reaction network with orientation}\label{sssec:dsCRN}

For the purpose of our simulation, we introduce and investigate two variations of the sCRN model: the \textbf{directed sCRN (d-sCRN)} and the \textbf{clockwise sCRN (c-sCRN)}. Directed sCRN provides each bimolecular reaction with a global direction $d\in\{\uw,\rw,\dw,\lw\}$. We write $(A,B,C,D,d)$ to indicate that, when $d=\uw$ (resp. $\rw,\dw,\lw$), two adjacent species $A$ and $B$ could turn into species $C$ and $D'$ simultaneously if $B$ is in the North (resp. East, South, West) of $A$, where ``North'' means the $(0,1)$ direction in $\Zbb^2$. Notice that $\dw,\lw$ are actually redundant since the bimolecular reactions are symmetric for all pairs of adjacent species; we use these notations just for convenience when describing protocols in Section~\ref{sec:sCRN-dsCRN}. We use $(A,B,C,D,d)^{-1}$ to represent the reverse reaction $(C,D,A,B,d)$. For a unimolecular reaction such that species $A$ can turn into $B$, we simply use the notation $(A,B,\odot)$. Clockwise sCRN, defined on the triangular lattice, provides each bimolecular reaction with a local direction $d\in\{0,\cdots,5\}$ such that $(A,B,C,D,d)$ represents the reaction where $B$ is in the $d$-th direction in the view of $A$. The only common knowledge is that for every species, $d$ increases in the clockwise order.

\subsubsection{Reachability and termination}\label{sssec:reachability}
We first define the \emph{one-step reachability}. For any two configurations $\alpha$ and $\beta$, we say that $\beta$ is reachable from $\alpha$ in one step if there exists a single cell or a pair of adjacent cells such that performing some $r\in R$ on these cells yields $\beta$. Write $\alpha\rw_\Gamma^1\beta$. 


Let $\rw_\Gamma$ be the reflexive transitive closure of $\rw^1_{\Gamma}$, a configuration $\beta$ is \emph{reachable} from another configuration $\alpha$ if $\alpha\rw_\Gamma\beta$. That is, $\beta$ is reachable from $\alpha$ in one or more steps. A configuration $\alpha$ is \emph{$\Gamma$-reachable} (or \emph{reachable} when $\Gamma$ is clear from the context) if $S\rw_\Gamma\alpha$, and write $\Acal(\Gamma)$ to denote the set of all $\Gamma$-reachable configurations. We say $\alpha$ is \emph{$\Gamma$-terminal} if $\alpha$ is $\Gamma$-reachable and there exists no configuration $\beta\neq\alpha$ which is $\Gamma$-reachable from $\alpha$. We denote the set of $\Gamma$-terminal configurations by $\Acal_*(\Gamma)$.

\subsection{Abstract tile assembly model}\label{ssec:aTAM}

In his Ph.D. thesis~\cite{Win98}, Winfree proposed the \textbf{abstract tile assembly model (aTAM)}. This model formalizes the self-assembly of molecules (such as DNA), describing the process wherein simple tiles spontaneously attach to each other to produce complex structures. In this paper, we use the model proposed in his later work with Rothemund~\cite{RW00}. Two differences we make are that we use the definition of assembly in~\cite{DL12} to align with the one used in Section~\ref{ssec:TA} when introducing tile automata system, and we use the definition of attachability in~\cite{ALQ01} to make our description of the protocol in Section~\ref{sec:aTAM} more comprehensive. Note that these two definitions are equivalent to the original one given in~\cite{Win98}. We include a brief description here to make this paper self-contained.

A \emph{tile} is an oriented unit square with the north, east, south, and west sides labeled from some alphabet $\Sigma$. For a tile $t$, define the \emph{state} $\sigma$ to be the 4-tuple $\sigma(t)=(\sigma_N(t), \sigma_E(t), \sigma_S(t), \sigma_W(t))$, consisting of the labels (also called the \emph{glues}) on its four sides. It is assumed that $\nullsf\in\Sigma$, and we use $\nullsf=(\nullsf,\nullsf,\nullsf,\nullsf)$ (abuse of notation) to represent the absence of any other tiles. A \emph{glue strength function} $g:\Sigma\times\Sigma\rw\Nbb\cup{0}$ maps a pair of glues to a natural number $\Nbb$.
\[g(x,y)=g(y,x)=\begin{cases}
\hat{g}(x), & x=y\\
0, & \text{otherwise}
\end{cases},\ \forall x,y\in\Sigma, \text{ for some }\hat{g}:\Sigma\rw\Nbb\]

Given a finite set of states $Q$, it is allowed that an infinite number of tiles of the same state occupying locations in $\Zbb^2$. A \emph{configuration} is a mapping from $\Zbb^2$ to $Q$ ($\nullsf\in Q$). Let $\dagQ=Q\setminus\{\nullsf\}$, we consider the partial function $\dagalpha$ from $\Zbb^2$ to $\dagQ$ where $\dagalpha(v)=\alpha(v)$ for all $v\in\Zbb^2$ s.t. $\alpha(v)\neq\nullsf$. We can view the domain of $\dagalpha$ as a subgraph of $\Zbb^2$ lattice. We call $\alpha$ an \emph{assembly} if $\dagalpha$ has a connected, non-empty domain $\Omega_\alpha\subset\Zbb^2$, and $\Omega_\alpha$ is called the \emph{shape} of this assembly $\alpha$.

Let $\alpha,\beta$ be two configurations s.t. $\dagalpha,\dagbeta$ have disjoint domains. Addition of configurations $\alpha,\beta$, denoted as $\gamma=\alpha+\beta$ (or $\beta=\gamma-\alpha$) is defined by 
\[\gamma=\begin{cases}
    \alpha(x,y), &\beta(x,y)=\nullsf.\\
    \beta(x,y), &\alpha(x,y)=\nullsf.\\
    \nullsf, &\alpha(x,y)=\beta(x,y)=\nullsf.
\end{cases}\]

The \emph{interaction strength} between two adjacent tiles with their abutting sides labeled $x,y$ is $g(x,y)$. Given a configuration $\alpha$, define the \emph{binding graph} $B_\alpha$ whose vertices are those tiles in $\dagQ$, with an edge of weight $w$ between two vertices if the interaction strength between them is $w$. For some \emph{temperature} $\tau>0$, an assembly $\alpha$ is \emph{$\tau$-stable} if the graph $B_\alpha$ has min-cut $\geq\tau$.

We use $t_{(x_0,y_0)}$ where $t$ is a tile to represent the configuration s.t.
\[t_{(x_0,y_0)}(x,y)=\begin{cases}
    \sigma(t), &(x,y)=(x_0,y_0).\\
    \nullsf, &(x,y)\neq(x_0,y_0).
\end{cases}\]
For a $\tau$-stable assembly $\alpha$, we say the position $(x_0,y_0)$ is \emph{attachable} in $\alpha$ if there exists an assembly $\beta$ s.t. $\beta=\alpha+t_{(x_0,y_0)}$ and
\begin{equation}\label{eq1}
    \begin{aligned}
    &g(\sigma_N(t),\sigma_S(\beta(x_0,y_0+1)))+g(\sigma_E(t),\sigma_W(\beta(x_0+1,y_0)))\\
    &+g(\sigma_S(t),\sigma_N(\beta(x_0,y_0-1)))+g(\sigma_W(t),\sigma_E(\beta(x_0-1,y_0)))\geq\tau
    \end{aligned}
\end{equation}
Notice that Equation(\ref{eq1}) is equivalent to requiring $\beta$ to be $\tau$-stable.

By the attachment of a single tile at position $(x,y)$, $\beta$ is reachable from $\alpha$ in one step. Write $\alpha\rw^1_{\Tcal}\beta$. The \emph{reachability} $\rw_\Gamma$ is defined the same as in Section~\ref{sssec:reachability}. Also, we follow the definition of terminal set $\Acal_*(\Gamma)$ in Section~\ref{sssec:reachability} with an additional requirement that for any $\alpha\in\Acal_*(\Gamma)$, $\alpha$ must be $\tau$-stable.

Such a system of aTAM is represented by a quadruple $\Gamma=(Q,S,g,\tau)$, where $Q,g,\tau$ are as above and $S$ is the \emph{seed configuration}. In this paper (also as suggested in~\cite{RW00}), we consider only the unit-seeded system where $S=s_{(0,0)}$ is the configuration s.t. $S(0,0)=s$ for some \emph{seed tile} $s$ and all other locations are $\nullsf$.

\subsection{Cellular automata}\label{ssec:CA}

The cellular automata was first designed in von Neumann's book~\cite{JVN66}. It's a dynamical system that use local interaction to perform complex global behavior. Cellular automata has been studied to model natural phenomena, and it is also computationally universal. In this paper, we use the definition of synchronous cellular automata in the survey~\cite{BNR20}, and use the definition of asynchronous cellular automata in~\cite{FG08}. We include a brief description here to make this paper self-contained.

A $d-$dimensional cellular automata, whose underlying topology is $\Zbb^d$, is specified by a triple $\Gamma=(Q,\Ncal,f)$
\begin{itemize}
    \item $Q$ is the finite state set.
    \item $\Ncal=(\vec{v_1},\cdots,\vec{v_N})$ is the neighborhood vector of $N$ distinct elements of $\Zbb^d$. Then the \emph{neighbors} of a cell at location $\vec{v}\in\Zbb^d$ are $\{\vec{v}+\vec{v_i}\}_{i=1}^N$.
    \item $f:Q^N\rw Q$ is the \emph{local rule} that computes the next state of a cell from the states of its neighbors. i.e. the next state of a cell $\vec{v}$ is $f(a_1,\cdots,a_N)$ where $a_i$ is the state of its neighbor $\vec{v}+\vec{v_i}$. All cells use the same rule.
\end{itemize}

The \emph{configuration} is a mapping $\alpha:\Zbb^d\rw Q$ that specifies the states of all cells. We denote by $\Ccal(d,Q)=Q^{\Zbb^d}$ the set of all possible configurations. In this paper, we only consider dimension $d\leq2$ and the \emph{von-Neumann} neighborhood containing all $\vec{v_i}$ s.t. $\onenorm{\vec{v_i}}\leq1$. In $2$-dimensional, this is $\Ncal=(\vec{v_O},\vec{v_N},\vec{v_E},\vec{v_S},\vec{v_W})$ where $\vec{v_O}=(0,0),\vec{v_N}=(0,1),\vec{v_E}=(1,0),\vec{v_S}=(0,-1),\vec{v_W}=(-1,0)$.

Another way to identify a cellular automata is by its \emph{global transition function} $G$. We define $G$ in different ways depending on whether the cells are updated simultaneously. We now introduce these two kinds.
\medskip\\
\textbf{synchronous CA (CA).} All cells are updated simultaneously. For this case, the global function $G:\Ccal(d,Q)\rw\Ccal(d,Q)$ maps configuration $\alpha$ to $\beta$ ($\beta=G(\alpha)$) if for all $\vec{v}\in\Zbb^d$, $\beta(\vec{v})=f(\alpha(\vec{v}+\vec{v_1}),\cdots,\alpha(\vec{v}+\vec{v_n}))$.
\medskip\\
\textbf{asynchronous CA (async-CA)}. In this paper we consider the \emph{fully asynchronous} updating scheme, under which a single cell is chosen at random at each time step. The global function $G:\Ccal(d,Q)\times\Zbb^d\rw\Ccal(d,Q)$ takes at input a configurations $\alpha$ and a randomly chosen cell $\vec{u}\in\Zbb^d$, and output the configuration $\beta=G(\alpha)$ s.t. \[\beta(\vec{v})=\begin{cases}
    f(\alpha(\vec{v}+\vec{v_1}),\cdots,\alpha(\vec{v}+\vec{v_N})), & \vec{v}=\vec{u}\\
    \alpha(\vec{v}), &\text{otherwise}
\end{cases}.\]

Moreover, in this paper we consider the \textbf{non-deterministic async-CA} whose local function $f$ could be \emph{nondeterministic}. That is, for any cell $v$, let $a_O,a_N,\cdots,a_W$ be the states of the neighborhood $\vec{v},\vec{v}+\vec{v_N},\cdots,\vec{v}+\vec{v_W}$ respectively. Then $f$ maps $(a_O,a_N,a_E,a_S,a_W)$ to a set of possible next states, from which one is picked non-deterministically at each time. A configuration $\beta$ is reachable in one step from another configuration $\alpha$, denoted by $\alpha\rw^1_\Gamma\beta$, if $\beta=G(\alpha)$. The \emph{reachability} $\rw_\Gamma$ is defined as in Section~\ref{sssec:reachability}, and the terminal set $\Acal_*(\Gamma)$ is a subset of the reachable set that is a fixed point of the global function $G$. i.e. $\Acal_*(\Gamma)=\{\alpha\in\Acal(\Gamma):G(\alpha)=\alpha\}$.

\subsection{Tile automata}\label{ssec:TA}
The tile automata model, combining \emph{2-handed assembly model} with local state-change rules between pairs of adjacent tiles, is a marriage between tile-based self-assembly and asynchronous cellular automata. It was first proposed by Chalk\etal~\cite{CL18} and we follow their model definition. Variations like affinity strengthening and unit-seeded system were considered in~\cite{CGS21} and~\cite{RDS23} respectively. In this paper, we restrict ourselves on the system with both constraints, and slightly modify the description of single tile attachment to align with the one used in aTAM model introduced in Section~\ref{ssec:aTAM}. We include a brief description here to make this paper self-contained.

Similar to how we define a tile in the tile-based self-assembly, here a \emph{tile} $t$ is a unit square located at $\Zbb^2$, each assigned a \emph{state} $\sigma(t)$ from the finite state set $Q$. Similar to aTAM, we use $\nullsf\in Q$ to represent the lack of any other tile. An \emph{affinity function} $g:Q^2\times D\rw\Nbb$, where $D=\{\perp,\vdash\}$, represents the \emph{affinity strength} between two states with relative orientation $d\in D$. To describe it explicitly, the affinity strength between two adjacent tiles $t_1,t_2$ located at $(x_1,y_1),(x_2,y_2)$ is $\begin{cases}
    g(\sigma(t_1),\sigma(t_2),\perp), &\text{if }y_1>y_2\\
    g(\sigma(t_2),\sigma(t_1),\perp), &\text{if }y_1<y_2\\
    g(\sigma(t_1),\sigma(t_2),\vdash), &\text{if }x_1<x_2\\
    g(\sigma(t_2),\sigma(t_1),\vdash), &\text{if }x_1>x_2
\end{cases}$.\\
And the affinity strength between a $\nullsf$ tile and any other tile is $0$.

A \emph{configuration} is a mapping $\alpha$ from $\Zbb^2$ to $Q$. Let $\dagQ=Q\setminus\{\nullsf\}$, we consider the partial function $\dagalpha:\Zbb^2\rw\dagQ$ like the one defined in Section~\ref{ssec:aTAM}. An \emph{assembly} is a configuration $\alpha$ where the domain of $\dagalpha$ being connected. The \emph{shape}, the additive notation ``$+$'', and the induced \emph{binding graph} $B_\alpha$ all follow the definitions in Section~\ref{ssec:aTAM}, where the weight $w$ on an edge $(t_1,t_2)$ now equals to the affinity strength between these two vertices. For a \emph{stability threshold} $\tau$, an assembly $\alpha$ is $\tau$-stable if $B_\alpha$ has min-cut $\geq\tau$.

In general, an assembly $\alpha$ could break into two pieces, (say, \emph{$\tau$-breakable}) if there's a cut in the binding graph with total affinity strength $<\tau$. On the other hand, two $\tau$-stable assembly could be combined along a border whose total strength sums to at least $\tau$. A special case is the \emph{single tile attachment}, which is similar to the tile assembly system. For a $\tau$-stable assembly $\alpha$, we say the position $(x_0,y_0)$ is \emph{attachable} in $\alpha$ if there exists $\tau$-stable assembly $\beta$ s.t. $\beta=\alpha+ t_{(x_0,y_0)}$. But in this paper we only consider the \emph{affinity-strengthening rule} s.t. for each transition rule that takes a state $\sigma_x^1$ to $\sigma_x^2$, it must satisfy that the affinity strength with all other states could only increase. In other words, for every $\sigma\in Q$, $g(\sigma_x^1,\sigma,d)\leq g(\sigma_x^2,\sigma,d)\ \forall d\in D$. Notice that with the affinity-strengthening rule, any tile that has attached to a configuration can not fall off. This implies that the whole assembly is unbreakable after any state transition.

Like asynchronous cellular automata, tile automata has a local state changing rule. The \emph{transition rule} $r$ is a 5-tuple $(\sigma_x^1,\sigma_y^1,\sigma_x^2,\sigma_y^2,d)$ where $\sigma_x^1,\sigma_y^1,\sigma_x^2,\sigma_y^2\in Q$ and $d\in D$. It means that if states $\sigma_x^1$ and $\sigma_y^1$ are adjacent with relative orientation $d$, they can simultaneously turn into state $\sigma_x^2,\sigma_y^2$ respectively. In this paper, we allow the transition rule to be nondeterministic, i.e. there may be two rules $(\sigma_x^1,\sigma_y^1,\sigma_x^2,\sigma_y^2,d)$ and $(\sigma_x^1,\sigma_y^1,\sigma_x^3,\sigma_y^3,d)$ s.t. $(\sigma_x^2,\sigma_y^2)\neq(\sigma_x^3,\sigma_y^3)$.

A \textbf{tile automata system (TA)} is a 6-tuple $\Gamma=(Q,S,I,g,R,\tau)$ where $Q$ is a finite set of states, $g$ is an affinity function, $S$ is the initial configuration, $R$ is a set of transition rules, and $\tau$ is the stability threshold. We assume that initially, there is a set of tiles whose states belong to the \emph{initial state} $I\subseteq Q$ that could be used in attachment. Any configuration $\beta$ that is reachable from $\alpha$ in one step (denoted as $\alpha\rw^1_\Gamma\beta$) if and only if $\beta$ is formed by applying some transition rule to $\alpha$ or by an attachment of a tile in $I$. The \emph{reachability} $\rw_\Gamma$ is defined the same as in Section~\ref{sssec:reachability}. Also, we follow the definition of terminal set $\Acal_*(\Gamma)$ in Section~\ref{sssec:reachability} with an additional requirement that for any $\alpha\in\Acal_*(\Gamma)$, $\alpha$ must be $\tau$-stable.

We restrict ourselves to the \textbf{unit-seeded TA with affinity-strengthening rule (s-as-TA)} in this paper. The term ``unit-seeded'' means that there exists a single seed tile $s$ s.t. $S=s_{(0,0)}$ and only single tile attachment is allowed.

\subsection{Amoebot}\label{ssec:Amoe}
The amoebot model was originally proposed by Derakhshandeh\etal~\cite{DDG14}. To simplify the simulation, we use the notions of movement and transition function given by Alumbaugh\etal~\cite{ADD19}, and propose a slightly different definition of configuration. Inspired by amoeba, it is a model of programmable matter where particles perform simple computation according to local information, and can move via contraction and expansion. We include a brief description here to make this paper self-contained.

The amoebot model is an abstract computational model of \emph{programmable matter} consisting of \emph{particles}, which are simple computational units that can move and bond to others and exchange information by these bonds. The underlying topology is a triangular lattice with nearest neighbor connectivity $G_\bigtriangleup=(V,E)$. $V$ represents all possible positions of a particle, and $E$ represents all possible movement and information transitions between particles.

Every particle is either \emph{contracted} (occupying a single node) or \emph{expanded} (occupying two adjacent nodes). Particles are \emph{anonymous}, but each edges leaving a particle $p$ is locally labeled so that a particle can uniquely identify each of them: The labeling starts with $0$ at an arbitrary outgoing edge leading to a node that is only adjacent to one of the nodes occupied by $p$, then increases in clockwise order around the particle. Notice that while they have a common clockwise chirality, they may have different \emph{orientations} in $\Ocal=\{0,1,2,3,4,5\}$ encoding their offsets for local direction $0$ from global direction $0$ (to the right). i.e. they don't have consensus in global orientation.

Every particle has a constant-size, shared, local memory which can be read and write by both itself and its neighbor for communications. More formally, we denote an amoebot model by $\Gamma=(\Phi,S,\Sigma,\delta)$. Each particle has a \emph{state} from a finite set $\Phi$. A particle $p$ communicates with an neighboring particle $q$ by placing a \emph{flag} from a finite alphabet $\Sigma$ on the edges leading to $q$, so that $q$ can read this flag.

Particles move through \emph{expansions} and \emph{contractions}: A contracted particle can \emph{expand} into an unoccupied adjacent node and become expanded. The \emph{head} of an expanded particle is the node it last expanded into and the other node is the \emph{tail}. An expanded particle can \emph{contract} to its head, performing a movement toward its head, or contract back to its tail. The direction of the edge labeled $0$ remains constant during movement.\footnote{When a particle expands, there could be two choices of edge to be labeled $0$, in this case, the edge ''away'' from the particle is labeled $0$.} Neighboring particles can perform \emph{handover} in one of the two ways:
\begin{itemize}
    \item A contracted particle $p$ \emph{push} an expanded neighbor $q$ by expanding into one of the nodes occupied by $q$.
    \item An expanded particle $p$ \emph{pull} a contracted particle $q$ by contracting, forcing $q$ to expand into the currently vacating node.
\end{itemize}
Define $M$ to be the set of all possible movements, $M=\{\idlesf\}\cup\{\expandsf_i:i\in0,\cdots,5\}\cup\{\contractsf_i:i\in0,\cdots,9\}\cup\{\handoversf_i:i\in0,\cdots,5\}$. $\idlesf$ means that a particle does not move, $\expandsf_i$ and $\contractsf_i$ means expand/ contract toward the local direction $i$, and $\handoversf_i$ means that a contracted particle $p$ expand toward its $i$-th direction, forcing an expanded $q$ to contract.

In the execution of an amoebot algorithm, particles progress by performing atomic actions. 
We may assume a standard asynchronous computation model, i.e. only one particle is activates at a time. Each time a particle $p$ is activate, it acts according to a \emph{transition function} $\delta:\Phi\times\Sigma^{10}\times D\rw\Pcal(\Phi\times\Sigma^{10}\times D\times M)$.
where $D$ denote its tail direction of an expanded particle. We use $\varepsilon\in D$ to represents the contracted particle. The transition function takes as input the state of a particle, all the flags it reads, and maps to a set of \emph{turns}. Where a turn specifies a state to transition into, flags to set, and a movement to execute.

For our simulation, we align the notation of \emph{configuration} with other models. For a particle $p$, let $\sigma(p)=(\phi,o,d,f_0,\cdots,f_9)$, where $\phi\in \Phi$ is the state of $p$, $o\in\Ocal$ is the orientation of $p$, $d\in D$ is its tail direction, and $f_i$ is the flag it place at the edge labeled $i$. Let $Q=\Phi\times\Ocal\times D\times\Sigma^{10}$, a configuration is a mapping $\alpha:G_\bigtriangleup\rw Q$. For two configuration $\alpha,\beta$, $\beta$ is \emph{reachable} from $\alpha$ in one step (denoted $\alpha\rw^1_\Gamma\beta$) if $\alpha$ can become $\beta$ after a single particle activation. Let $S$ be the initial configuration, which consists of contracted particles forming a connected shape. According to what we have defined so far, we can also represent an amoebot system by $\Gamma=(Q,S,\delta)$, and we let $Q$ be the new ``state set'' of the amoebot system throughout the rest of this paper.

%% file: simdef_v6.tex
\section{Simulation}\label{sec:sim}
Our simulation definition is adapted from the one used in~\cite{HP23} and~\cite{ADD19}. Let $\Scal,\Tcal\in\{\textbf{sCRN, aTAM, s-as-TA, async-CA, amoebot}\}$ be two system in different models.

\subsection{Representation function}
Let $Q_\Scal,Q_\Tcal$ be the state sets of $\Scal,\Tcal$ respectively. A \emph{state representation function} $\Rcal$ from $Q_\Scal$ to $Q_\Tcal$ is a function which takes as input a state of $\Scal$, and returns either information about the state of $\Tcal$ or $\UNDsf$ (which implies the image of $\Rcal$ on this state is undefined).

Let $\mathbf{v}$ be the coordinates of a node in the underlying lattice $L$. Let $\alpha$ be a configuration of $\Scal$, and let $\alpha(\mathbf{v})$ be the state of $\alpha$ on the node $\mathbf{v}$.  A \emph{representation function} $\Rcal^*$ from $\Scal$ to $\Tcal$ takes as input an entire configuration $\alpha$ of $\Scal$ and apply $\Rcal$ to every state of $\alpha$, and returns either a corresponding configuration of $\Tcal$, or $\UNDsf$ if there is any state contained in $\alpha$ that mapped to $\UNDsf$.

Note that our definition is more relax than that in~\cite{HP23} and~\cite{ADD19}. In particular we allow $\UNDsf$ to be the image of $R$, which is motivated by the simulation definition in~\cite{AAFJ08}. The relaxation is necessary because when simulating sCRN by non-deterministic async-CA, it is possible that the sCRN perform a bimolecular reaction and hence two adjacent states $A,B$ change simultaneously. While in the non-deterministic async-CA, it can be one cell performing state transition at a time. Suppose $A$ change its state first, so the other cell $B$ must be mapped to different states before and after $A$ performing state transition. Another way of relaxing this definition to fit in with the limitation of non-deterministic async-CA is to let the representation function $\Rcal$ map a neighborhood of CA to a state in sCRN, but this definition is rather non-intuitive in other models. The details are given in Section~\ref{sec:CA}. Notice that except for the simulation of sCRN by non-deterministic acync-CA and the simulation between sCRN and amoebot system, we don't need this relaxation.

\subsection{Simulation Definition}\label{sec:simdef}
Roughly speaking, if we say that $\Scal$ simulates $\Tcal$, we want the pattern of $\Scal$ to evolve like that in $\Tcal$. Intuitively, we want for any sequence of configuration changes $\alpha\rw\beta$ in $\Tcal$, there exists a sequence of configuration changes $\alpha'\rw\beta'$ s.t. $\alpha',\beta'$ are mapped to $\alpha,\beta$ respectively. On the other hand, for any $\alpha'\rw\beta'$ in $\Tcal$, it is also required that there is a sequence of configurations that links the images of $\alpha',\beta'$ under $R^*$. i.e. $R^*(\alpha')\rw R^*(\beta')$.

\begin{definition}\label{def:equiv}
We say that $\Scal$ and $\Tcal$ have equivalent productions (under $\Rcal$), and we write $\Scal\Lrw_\Rcal\Tcal$ if the following conditions hold:
\begin{enumerate}
    \item $\{\Rcal^*(\alpha')|\alpha'\in\Acal(\Scal)\}=\Acal(\Tcal)\cup\{\UNDsf\}$.
    \item $\{\Rcal^*(\alpha')|\alpha'\in\Acal_*(\Scal)\}=\Acal_*(\Tcal)$.
\end{enumerate}
\end{definition}

\begin{definition}\label{def:flw}
We say that $\Tcal$ follows  $\Scal$ (under $\Rcal$) and we write $\Tcal\flw_\Rcal\Scal$ if, $\alpha'\rw_{\Scal}\beta'$ for some $\alpha',\beta'\in\Acal(\Scal)$ implies that $\Rcal^*(\alpha')\rw_{\Tcal}\Rcal^*(\beta')$ or, either $\Rcal^*(\alpha')=\UNDsf$ or $\Rcal^*(\beta')=\UNDsf$.
\end{definition}

To be more rigorous, for every configuration $\alpha\in\Tcal$, there must exist a reachable set in $\Scal$ which is mapped to $\alpha$ s.t. all the images (under $R^*$) of configurations that can grow from this set together cover all possible next configurations from $\alpha$.

\begin{definition}\label{def:mdl}
We say that $\Scal$ models  $\Tcal$ (under $\Rcal$) and we write $\Scal\models_\Rcal\Tcal$, if for every $\alpha\in\Acal(\Tcal)$, there exists $\Pi\subset\Acal(\Scal)$ where $\Pi\neq\emptyset$ and $\Rcal^*(\alpha')=\alpha\ \forall\alpha'\in\Pi$, such that for every $\beta\in\Acal(\Tcal)$ where $\alpha\rw\beta$, the followings hold:
\begin{enumerate}
    \item for every $\alpha'\in\Pi$, there exists $\beta'\in\Acal(\Scal)$ where $\Rcal^*(\beta')=\beta$ and $\alpha'\rw_{\Scal}\beta'$.
    \item for every $\alpha''\in\Acal(\Scal)$ where $\alpha''\rw_{\Scal}\beta'$ s.t. $\Rcal^*(\alpha'')=\alpha$ and $\Rcal^*(\beta')=\beta$, there exists $\alpha'\in\Pi$ such that $\alpha'\rw_{\Scal}\alpha''$.
\end{enumerate}
\end{definition}
\begin{definition}\label{def:sim}
    We say that $\Scal$ simulates $\Tcal$ (under $\Rcal$) if $\Scal\Lrw_\Rcal\Tcal$, $\Tcal\flw_\Rcal\Scal$, and j$\Scal\models_\Rcal\Tcal$.
\end{definition}

\input{nota}

%% file: nota.tex
\subsection{Some notations}
In the succeeding chapters, we are going to state our simulation results. Before continuing, we give some notations that will be used frequently.
\begin{itemize}
    \item For $d=N,E,S,W$, define $d^{-1}=S,W,N,E$ respectively.
    \item Define $\Dcal=\{V,E,S,W\}$ and $\psi_\Dcal=(\psi_N,\psi_E,\psi_S,\psi_W)$.
    \item Define $\psi_{-d}=\{\psi_N,\psi_E,\psi_S,\psi_W\}\setminus\{\psi_d\}$. That is, we use $-d$ to represent the set of directions without $d$.
    \item $\lrangle{\psi',\psi_{-d}}$ means that we place $\psi'\cup\psi_{-d}$ in the $(N,E,S,W)$ order.
\begin{itemize}
    \item If $d=N$, $\lrangle{\psi',\psi_{-d}}=(\psi',\psi_E,\psi_S,\psi_W)$.
    \item If $d=E$, $\lrangle{\psi',\psi_{-d}}=(\psi_N,\psi',\psi_S,\psi_W)$.
    \item If $d=S$, $\lrangle{\psi',\psi_{-d}}=(\psi_N,\psi_E,\psi',\psi_W)$.
    \item If $d=W$, $\lrangle{\psi',\psi_{-d}}=(\psi_N,\psi_E,\psi_S,\psi')$.
\end{itemize}
\end{itemize}

%% file: s-sCRN_to_s-d-sCRN_v9.tex
\section{Simulation of unit-seeded directed sCRN by unit-seeded sCRN (up to rotation and reflection)}\label{sec:sCRN-dsCRN}
In this section, we show that an unit-seeded sCRN can simulate an unit-seeded directed sCRN up to rotation and reflection.
\begin{theorem}\label{thm:sCRN-dsCRN}
    Given a unit-seeded directed sCRN $\Gamma=(Q,S,R)$, there exists a unit-seeded sCRN $\Gamma'=(Q',S',R')$ which simulates $\Gamma$ up to rotation and reflection.
\end{theorem}
\subsection{Simulation overview}
Given a s-d-sCRN $\Gamma=(Q,R,S)$, we want to construct a s-sCRN $\Gamma'=(Q',R',S')$ to simulate $\Gamma$. For bimolecular reactions with a specified direction $d$, it is required that each species has a common knowledge of a global orientation. We achieve this by using a 2-hop coloring. The main idea is to partition the plane into blocks of $9$ cells, and color a block with ordering numbers $1,\cdots,9$ if and only if a reaction is about to happen on it. The simulation consists of the following three parts:
\begin{itemize}
    \item \ptc{Determine\_Global\_Orientation}: In this protocol, we color the seed and the eight cells around it with $\{1,\cdots,9\}$.
    \item \ptc{Growing\_Reactions}: If a reaction is possible on a specific cell, we first color the eight cells around it to make sure the global orientation is clear for that cell, and then perform the reaction.
    \item \ptc{State\_Transitions}: Simulate any reactions at a cell who has already known the global orientation.
\end{itemize}

To dive into the details, we first define the \emph{block} of a cell $u$, $\Bcal(u)=\{v:\maxnorm{v-u}\leq1\}$, to be the set of $9$ cells around $u$. And sometimes we use the block of a species $\rho$, write $\Bcal(\rho)$ (with abuse of notation), to mean the block of the position of $\rho$. Define $\dagQ=Q\setminus\{\Ocal\}$. For any reaction of the form $(A,\Ocal,B,C,d)$ where $A,B,C\in\dagQ$, we call it the \emph{growing reaction}, and we denote the set of growing reactions in $R$ by $R_g$.

We now sketch the idea of our simulation. Let $S'=S$ where $S$ maps every cell to the blank state $\Ocal$ except for a special cell that mapped to a seed state, say $s$. In the first part \texttt{Determine\_Global\_Orientation}, we color $\Bcal(s)$ with $\{1,\cdots,9\}$ to give an orientation to the system. And in the second part \texttt{Growing\_Reactions}, we ensure that in the simulation process of any growing reaction $(A,\Ocal,B,C,d)$, the state transition can be performed only if a coloring has been given to $\Bcal(\Ocal)$. This provides the information of the determined global orientation for the simulation of future reactions. The last part \texttt{State\_Transitions} gives the reactions we need for simulating reactions in a complete coloring configuration, which is relatively naive. We construct $\Gamma'$ by adding states and reactions to $Q'$ and $R'$ successively. The full description of the protocols and more details are provided in the following paragraphs.

\paragraph{\texttt{Determine\_Global\_Orientation}}

This protocol aims to color all cells in the block of the seed. So that the coloring can be extended to the entire surface and be used as the global orientation. Let $S'=S$ where $s$ is the special seed state in $\Gamma$. We color $\Bcal(s)$ with $\{1,\cdots,9\}$ and use $\Ocal^i$ to denote the colored cell, where the superscript $i\in\{1,\cdots,9\}$ is the color it receives.

Initially, we add the reaction $s+\Ocal\rw0+1$ to $R'$. Then we build the block of $0$ by forming $X,Y$ toward each direction of $0,1$ respectively. When $X,Y$ meet each other, they perform state transition to $(2,3)$ respectively, meaning that they are parallel and adjacent to $(0,1)$. From $2$ we grow $Z$ toward the remaining two directions, the first to meet $X$ perform $Z+X\rw4+5$, and the other $Z$ meeting $5$ begins to color the block of $0$ in a counter clockwise order. Except for $s$, which is to be turned into $s^5$, other cells receive $\Ocal^i$ where $i$ depends on their positions. After the coloring is complete, we need to clear off the redundant states $Y,Z$ grown outside $\Bcal(s^5)$. So we also have $\Ocal^8+Y\rw\Ocal^8+\Ocal$, $\Ocal^6+Z\rw\Ocal^6+\Ocal$, and $\Ocal^4+Z\rw\Ocal^4+\Ocal$ in $R'$. In total we must have the following reactions in $R'$:
\begin{enumerate}
    \item $s+\Ocal\rw0+1$.\comm{Initiate the coloring.}
    \item $0+\Ocal\rw0+X$, $1+\Ocal\rw1+Y$, $X+Y\rw2+3$.
    \item $2+\Ocal\rw2+Z$, $Z+X\rw4+5$, $Z+5\rw\Ocal^1+\Ocal^2$.
    \item $\Ocal^2+4\rw\Ocal^2+\Ocal^3$, $\Ocal^3+2\rw\Ocal^3+\Ocal^6$, $\Ocal^6+3\rw\Ocal^6+\Ocal^9$, $\Ocal^9+1\rw\Ocal^9+\Ocal^8$, $\Ocal^8+3\rw\Ocal^8+\Ocal^7$, $\Ocal^7+2\rw\Ocal^7+\Ocal^4$, $\Ocal^4+0\rw\Ocal^4+s^5$.
    
    \comm{$s^5$ is the new seed with its block colored.}
    \item $\Ocal^8+Y\rw\Ocal^8+\Ocal$, $\Ocal^6+Z\rw\Ocal^6+\Ocal$, $\Ocal^4+Z\rw\Ocal^4+\Ocal$.
    
    \comm{Clear off the rubbish.}
\end{enumerate}
Figure~\ref{fig:determine_orientation} shows an evolution when we have this set of reactions.

Notice that the simulation system must be indifferent up to rotation and reflection since any starting reaction $s+\Ocal\rw A+B$ looks the same for $s$ in its four directions, and then any reaction can be perform symmetrically around $\overline{AB}$. Therefore, ``simulation up to rotation and reflection'' is a necessary relaxation.

\begin{figure}[htbp]
    \centering
    \includegraphics[width=\textwidth]{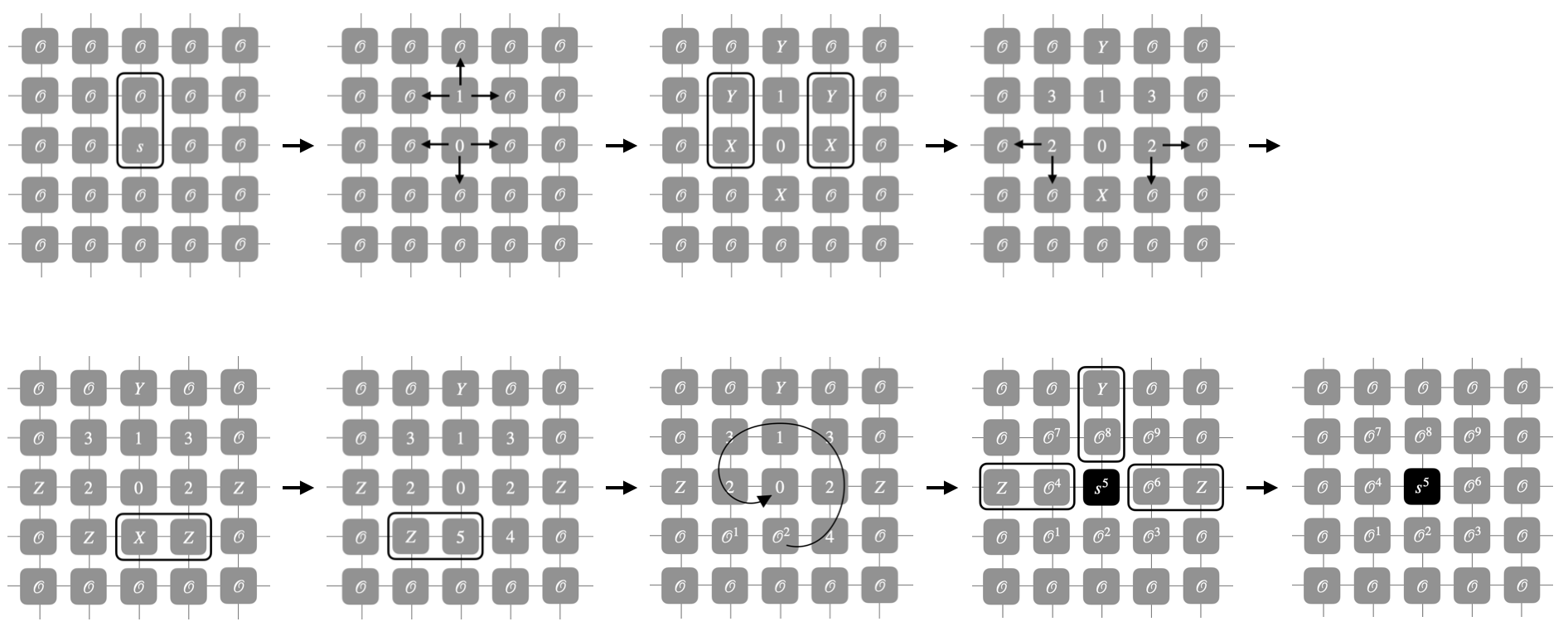}
    \caption{Coloring the block of the seed by} \ptc{Determine\_Global\_Orientation}
    \label{fig:determine_orientation}
\end{figure}

\paragraph{\texttt{Growing\_Reactions}}
For any growing reactions of this form $(A,\Ocal,B,C,d)$, we first give a coloring on $\Bcal(\Ocal)$ and then perform the reaction. Before giving the protocol, there are some notations we need to define first. We use $\psi^i$ to represent the colored species, where the superscript $i$ is the color of that species. In the simulation, we will have $8$ kinds of states in $Q'$, which are introduced in detail later:
\begin{align*}
    & \{\Ocal\}, \{\Ocal^i:i\in[9]\}, \{\sigma^i:\sigma\in\dagQ,i\in[9]\}, \{\chi^i:i\in[9]\}, \{\chi_{ij}:i,j\in[9]\},\\
    & \{\chi_i^j:i,j\in[9]\}, \{\chi_i^jp:i,j\in[9]\}, \{\chi_i^jb_0b_1:i,j\in[9], b_0,b_1\in\{0,1\}\}.
\end{align*}
Let $Q'$ be the union of them. Among these states, let 
\[\Xi^j=\{\sigma^j,\Ocal^j,\chi^j,\chi_i^j,\chi_i^jp,\chi_i^jb_0b_1:\sigma\in\dagQ,i\in[9],b_0,b_1\in\{0,1\}\}\]
be the set of all colored states whose color is $j$. And let $\Xi=\cup_{j=1}^9\Xi^j$.

Let $\dagQ_{\text{colored}}=\{\sigma^i:\sigma\in\dagQ,i\in[9]\}$. Through the representation function $\Rcal$, we map the state $\psi^i\in\dagQ_{\text{colored}}$ to $\psi$ for all $i\in[9]$. Except for the states in $\dagQ_{\text{colored}}$, all other states are mapped to the blank state $\Ocal$. We say that a configuration $\alpha'\in\Gamma'$ satisfies the \emph{complete coloring property} if for any species $\psi\in\dagQ_{\text{colored}}$, $\Bcal(\psi)$ are all colored. This protocol aims at maintaining this property.

Suppose now we have a configuration $\alpha'\in\Gamma'$ that satisfies the complete coloring property, and there is some growing reaction $(A,\Ocal,B,C,d)\in R$ that can be performed on $\alpha=\Rcal^*(\alpha')$. 

For example, Figure~\ref{fig:complete_coloring} shows a configuration with the complete coloring property. If $(A,\Ocal,B,C,\rw)\in R$, Then we want $A^1,\Ocal^2$ to be able to become $(B^1,C^2)$ while preserving the complete coloring property at the same time. Existing colors $1,2$ on $A,\Ocal$ help specifying the growing direction. Since $\Bcal(A^1)$ has been colored, we only need to color $W_{A^1}(\Ocal^2)=\Bcal(\Ocal^2)\setminus\Bcal(A^1)$, which are the three vertical cells in Figure~\ref{fig:complete_coloring} enclosed by the dashed lines. We use the notation $W_\psi(\psi')$ to represent $\Bcal(\psi')\setminus\Bcal(\psi)$ for adjacent species $\psi,\psi'\in Q'$.

Therefore, the first reaction we add to $R'$ is $A^1+\Ocal^2\rw A^1+\chi_1^2$. $\chi_1^2$ serves as an signal to begin the local coloring process. For the remaining part of this section, we'll give several examples of configuration around the block of $\chi_1^2$, and then specify the required reactions according to that situation.

Since $\Bcal(A^1)$ must have been completely colored due to the complete coloring property, we deliver the information of growing direction to the East cell of $\chi_1^2$. Suppose the cell is at position $v$ with state $\psi(v)$. To classify the state possibly encountered by $\chi_1^2$, we need to use the following property of a special kind of species $\{\chi^i,i\in[9]\}$: As soon as some state $\chi^i$ ($i\in[9]$) appears, the block of $\chi^i$ must have been colored.

With this property, we divide the possible states of $v$ into $4$ classes:
\begin{enumerate}
    \item $\psi(v)=\sigma^3$ for some $\sigma\in\dagQ\cup\{\chi\}$. It means that $\Bcal(v)$ has been colored, which implies $\Bcal(\chi_1^2)\setminus\Bcal(A^1)\subset\Bcal(v)$ has been colored as well. Then we directly turn $\chi_1^2$ to $\chi^2$, representing that $\Bcal(\chi_1^2)$ has been colored. Therefore we add the following reactions to $R'$:
    \begin{flalign*}
        \chi_1^2+\psi^3\rw\chi^2+\psi^3, \text{ for all }\psi\in\dagQ\cup\{\chi\}.&&
    \end{flalign*}
        
    \item $\psi(v)=\Ocal\text{ or }\Ocal^3$. We don't know if $\Bcal(v)$ is colored, therefore we add the reaction $\chi_1^2+\Ocal/\Ocal^3\rw\chi_1^2p+\chi_2^300$. This announce the starting point of coloring $\Bcal(\chi_2^300)\setminus\Bcal(\chi_1^2p)$. Therefore we add the following reactions to $R'$:
    \begin{flalign*}
        &\chi_1^2+\Ocal\rw\chi_1^2p+\chi_2^300.\\
        &\chi_1^2+\Ocal^3\rw\chi_1^2p+\chi_2^300.&&
    \end{flalign*}
        
    \item $\psi(v)=\chi_6^3\text{ or }\chi_9^3$. We need to be cautious about the deadlock in Figure~\ref{fig:deadlock}. In this case, we have $\chi_1^2$ continue to color the block of itself. Therefore we add the following reactions to $R'$:
    \begin{flalign*}
        &\chi_1^2+\chi_6^3\rw\chi_1^2p+\chi_2^300.\\
        &\chi_1^2+\chi_9^3\rw\chi_1^2p+\chi_2^300.&&
    \end{flalign*}

    \item $\psi(v)=\chi_1^3$. It means that some species $\sigma^1$ (where $\sigma\in\dagQ$) located two cells away from $A^1$ is also attempting to perform a growing reactions in an opposite direction. i.e. they are growing toward each other. Again by the complete coloring assumption, $\Bcal(\chi_1^2)\setminus\Bcal(A^1)\subset\Bcal(\sigma^1)$ and $\Bcal(\chi_1^3)\setminus\Bcal(\sigma^1)\subset\Bcal(A^1)$. Both $\chi_1^2,\chi_1^3$ know their blocks are colored, so they turn into $\chi^2$ and $\chi^3$ simultaneously. Therefore wee add the following reaction to $R'$:
    \begin{flalign*}
        \chi_1^2+\chi_1^3\rw\chi^2+\chi^3.&&
    \end{flalign*}
\end{enumerate} 

For case $2$ and $3$, the coloring of $\Bcal(\chi_1^2p)$ has not been done, $\chi_2^300$ is just produced to start coloring the cells to the North and South, call them $v_N,v_S$ respectively. $\chi_2^300$ will be turned into $\chi_2^3b_0b_1$ for some $b_0,b_1\in\{0,1\}$ that observes and records whether $v_N,v_S$ has been colored. We call $\chi_2^3b_0b_1$ the \emph{coloring species}. Let the states of $v_N,v_S$ be $\psi(v_N),\psi(v_S)$. For convenience, we further let the cell to the East of $\chi_2^300$ be $v_E$ with state $\psi(v_E)$. There are several situations that probably happens:

\begin{enumerate}
    \item $\psi(v_N)=\Ocal$. $\chi_2^30b$ does not know whether it encounter $v_N$ or $v_E$, so it turns the cell into $\chi_{23}$. $\chi_{23}$ then starts observing its neighbors, when seeing $\psi\in\Xi^5$, it confirms that itself must be colored $6$. Therefore we add the following reactions to $R'$:
    \begin{flalign*}
        &\chi_2^30b+\Ocal\rw\chi_2^30b+\chi_{23}, \text{ for all }b\in\{0,1\}.&&\\
        &\chi_{23}+\psi\rw\Ocal^6+\psi, \text{ for all }\psi\in\Xi^5.&&
    \end{flalign*}
    The case $\psi(v_S)=\Ocal$ is similar, so we add the following reactions to $R'$:
    \begin{flalign*}
        &\chi_2^3b0+\Ocal\rw\chi_2^3b0+\chi_{23}, \text{ for all }b\in\{0,1\}.&&\\
        &\chi_{23}+\psi\rw\Ocal^9+\psi, \text{ for all }\psi\in\Xi^8.&&
    \end{flalign*}

    Notice that there might be a redundant state $\chi_{23}$ produced at cell $v_E$, which must be turned back into $\Ocal$ later.

    \item $\psi(v_N)\in\{\chi_{39},\chi_{45},\chi_{54}\}$. When encountering these states, $\chi_2^30b$ knows by the index that they are redundant states from other simulation process of some growing reactions, and that the cell it sees is $v_N$. We can directly turn that cell into $\Ocal^6$. Therefore we add the following reactions to $R'$:
    \begin{flalign*}
        \chi_2^30b+\psi\rw\chi_2^31b+\Ocal^6, \text{ for all }\psi\in\{\chi_{39},\chi_{45},\chi_{54}\}.&&
    \end{flalign*}
    The case $\psi(v_S)\in\{\chi_{36},\chi_{78},\chi_{87}\}$ is similar, so we add the following reactions to $R'$:
    \begin{flalign*}
        \chi_2^3b0+\psi\rw\chi_2^3b1+\Ocal^9, \text{ for all }\psi\in\{\chi_{36},\chi_{78},\chi_{87}\}.&&
    \end{flalign*}

    \item $\psi(v_N)=\chi_5^6b0$. This means that there is some species $\sigma^4$ (where $\sigma\in\dagQ$) performing some growing reactions toward $v_N$, so $\chi_2^30b$ and $\chi_5^6b0$ could help each other complete their coloring process by turning itself into a colored state. Therefore we add the following reaction to $R'$:
    \begin{flalign*}
        \chi_2^30b+\chi_5^6b0\rw\chi_2^31b+\chi_5^6b1, \text{ for }b=0,1.&&
    \end{flalign*}
    For the other direction, $\psi(v_N)$ could be $\chi_4^60b$, indicating that a species $\sigma^5$ is performing a growing reaction toward $v_N$. Similar as above, we must have 
    \begin{flalign*}
        \chi_2^30b+\chi_4^60b\rw\chi_2^31b+\chi_4^61b, \text{ for }b=0,1.&&
    \end{flalign*}
    The case $\psi(v_S)=\chi_8^90b\text{ or }\chi_7^9b0$ is symmetric, so we add the following reactions to $R'$:
    \begin{flalign*}
        &\chi_2^3b0+\chi_8^90b\rw\chi_2^3b1+\chi_8^91b, \text{ for }b=0,1.\\
        &\chi_2^3b0+\chi_7^9b0\rw\chi_2^3b1+\chi_7^9b1, \text{ for }b=0,1.&&
    \end{flalign*}

    \item $\psi(v_N)\in\Xi^6\setminus\{\chi_5^6b0,\chi_4^60b\}$. Then $\psi(v_N)$ is viewed as a colored species, so we directly turned $\chi_2^30b$ into $\chi_2^31b$. We add the following reactions o $R'$:
    \begin{flalign*}
        \chi_2^30b+\psi\rw\chi_2^31b+\psi, \text{ for all }\psi(v_N)\in\Xi^6\setminus\{\chi_5^6b0,\chi_4^60b\}.&&
    \end{flalign*}
    The case $\psi(v_S)\in\Xi^9\setminus\{\chi_8^90b,\chi_7^9b0\}$ is similar, so we add the following reactions to $R'$:
    \begin{flalign*}
        \chi_2^3b0+\psi\rw\chi_2^3b1+\psi, \text{ for all }\psi(v_S)\in\Xi^9\setminus\{\chi_8^90b,\chi_7^9b0\}.&&
    \end{flalign*}
\end{enumerate}

Eventually, $\chi_2^300$ will become $\chi_2^311$, representing that both $v_N,v_S$ have been colored. Then we use a reaction $\chi_2^311+\chi_1^2p\rw\Ocal^3+\chi^2$ to announce the termination of the local coloring process. The state transition can now be simulated by adding reaction $A^1+\chi^2\rw B^1+C^2$ to $R'$. One thing remains is to eliminate the redundant state $\chi_{23}$ produced in the first case. So we need the reaction $\Ocal^3+\chi_{23}\rw\Ocal^3+\Ocal$.

The above description shows the simulation process of a special case that apply a growing reaction to a species $A$ colored by $1$ toward East. Figure~\ref{fig:growing_reaction} gives a sequence of reactions that could possibly be performed in our simulation process when there is a growing reaction $(A,\Ocal,B,C,\rw)$ in $R$.

\begin{figure}[htbp]
    \centering
    \begin{subfigure}{0.65\textwidth}
        \centering
        \includegraphics[width=0.75\textwidth]{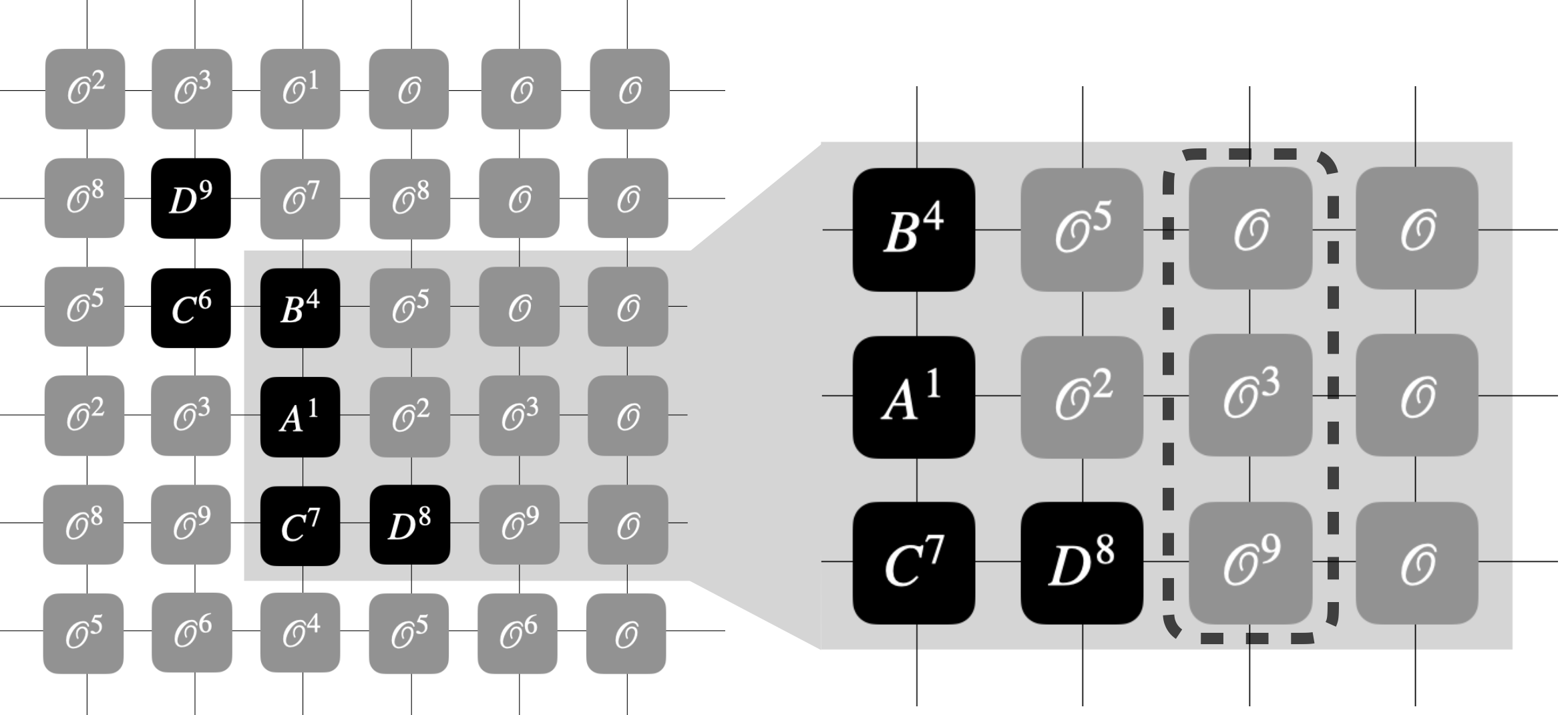}
        \caption{A complete coloring configuration.}
        \medskip
        \small The wall centered at $\Ocal^3$ are to be colored.
        \label{fig:complete_coloring}
    \end{subfigure}%
    \begin{subfigure}{0.35\textwidth}
        \centering
        \includegraphics[width=0.7\textwidth]{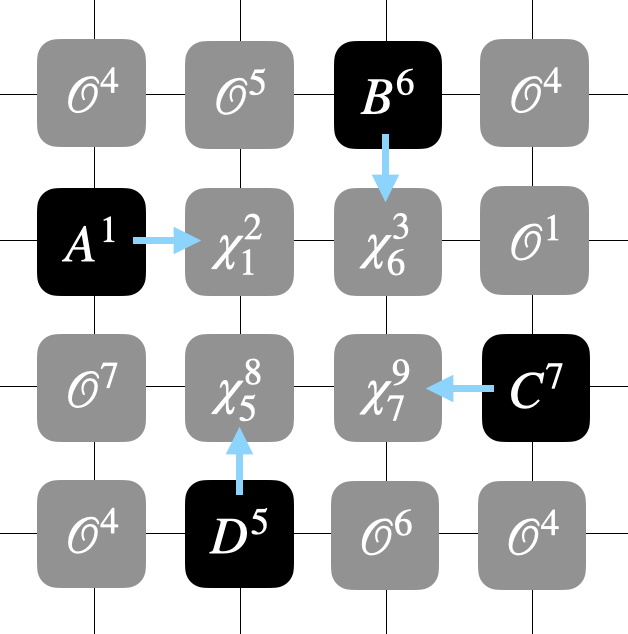}
        \caption{The undesired deadlock.}
        \label{fig:deadlock}
    \end{subfigure}%

    \bigskip
    \begin{subfigure}[b]{\textwidth}
        \centering
        \includegraphics[width=\textwidth]{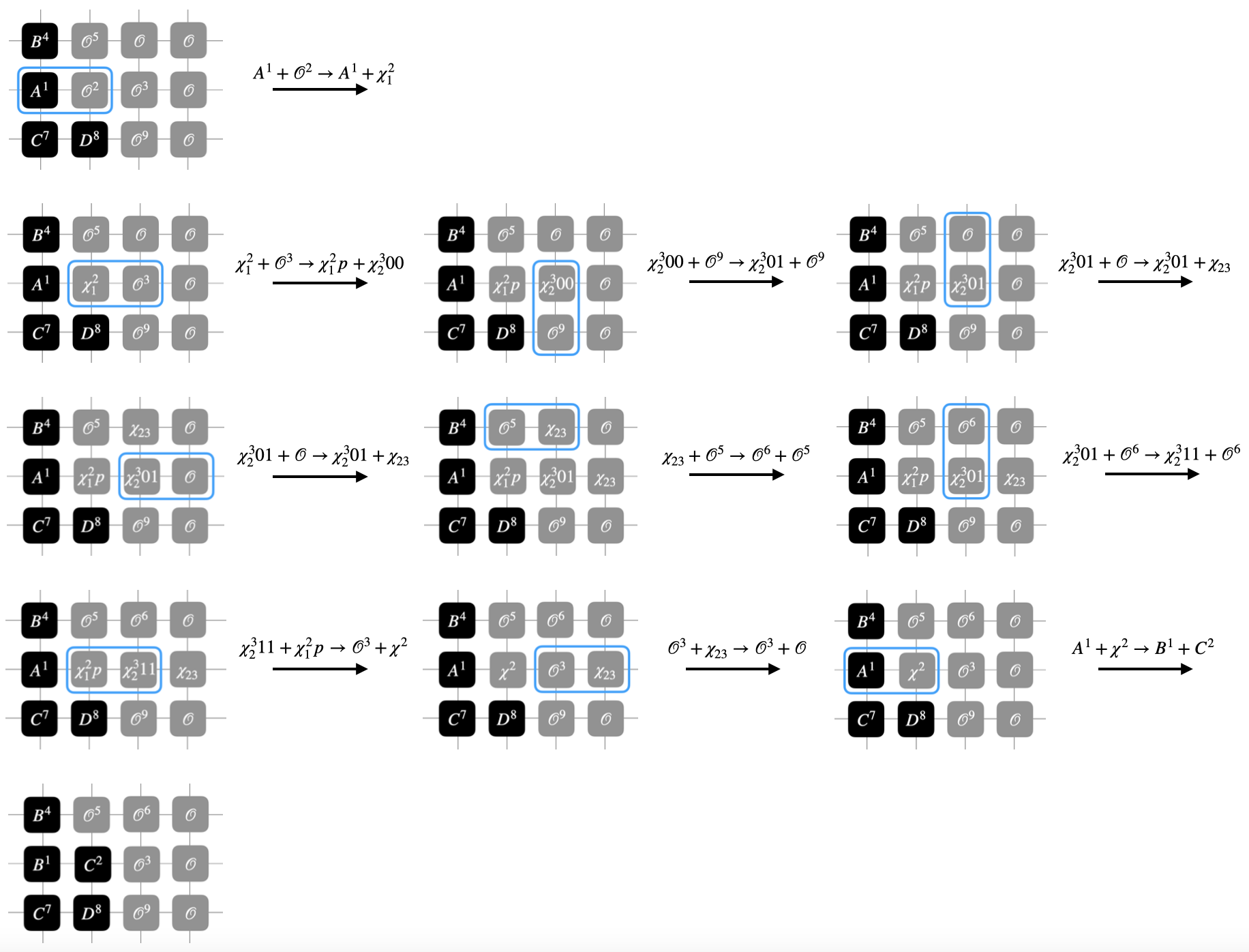}
        \caption{A possible sequence of reactions simulating $(A,\Ocal,B,C,d)$.}
        \label{fig:growing_reaction}
    \end{subfigure}%
    \caption{An example of the simulation of a growing reaction.}
     
    \label{fig:coloring_sequence}
\end{figure}


We could use the same method to construct the reaction sets needed for the growing reactions on a species colored $i\in[9]$ toward direction $d\in\{\uw,\rw,\dw,\lw\}$. This is just applying some permutations to all the reactions we constructed so far. 

Let the set be $\hat{R}$, we mean by ''apply a permutation $\pi$ to $\hat{R}$'' to represent changing all the subscript $i$ and superscript $j$ in the above simulation process to $\pi(i)$ and $\pi(j)$ respectively. In general, for every growing directions $d\in\Dcal$, we have to include the reactions on all $A^i,i\in[9]$. Observe that each case can be view as applying some rotations and translations to $\hat{R}$. For example, let $\pi=\begin{pmatrix}
    1&2&3&4&5&6&7&8&9\\
    3&6&9&2&5&8&1&4&7
\end{pmatrix}$. If we apply $\pi$ to $\hat{R}$, then we get all the reactions we need for simulating the growing reaction $(A,\Ocal,B,C,\uw)$ (toward North), starting at $A^3$. Further, let $\pi'=\begin{pmatrix}
    1&2&3&4&5&6&7&8&9\\
    3&1&2&6&4&5&9&7&8
\end{pmatrix}$, then applying $\pi'\circ\pi$ to $\hat{R}$ gives us the reactions needed for simulating the growing reaction toward North that starts at $A^2$.

We represent all the needed permutations by two-line notations.
 
First we define $4$ translations $\pi_r,\pi_l,\pi_u,\pi_d$:
\begin{equation*}
    \begin{split}
        \pi_r=\begin{pmatrix}
    1&2&3&4&5&6&7&8&9\\
    2&3&1&5&6&4&8&9&7
    \end{pmatrix}, \pi_l=\begin{pmatrix}
    1&2&3&4&5&6&7&8&9\\
    3&1&2&6&4&5&9&7&8
    \end{pmatrix},\\
    \pi_u=\begin{pmatrix}
    1&2&3&4&5&6&7&8&9\\
    4&5&6&7&8&9&1&2&3
    \end{pmatrix}, \pi_d=\begin{pmatrix}
    1&2&3&4&5&6&7&8&9\\
    7&8&9&1&2&3&4&5&6
    \end{pmatrix}.
    \end{split}
\end{equation*}
Let $\Pi_{\text{translation}}=\{id,\pi_r,\pi_l,\pi_u,\pi_d,\pi_u\circ\pi_r,\pi_u\circ\pi_l,\pi_d\circ\pi_r,\pi_d\circ\pi_l\}$. By applying each permutation in $\Pi_{\text{translation}}$ to $\hat{R}$, we could cover all the growing reactions $(A,\Ocal,B,C,\rw)$ starting at any $A^i\in Q'\text{ s.t. }i\in[9]$.

We also define $4$ rotations $\pi_1,\pi_2,\pi_3,\pi_4$:
\begin{equation*}
    \begin{split}
        \pi_1=id=\begin{pmatrix}
    1&2&3&4&5&6&7&8&9\\
    1&2&3&4&5&6&7&8&9
    \end{pmatrix}, \pi_2=\begin{pmatrix}
    1&2&3&4&5&6&7&8&9\\
    3&6&9&2&5&8&1&4&7
    \end{pmatrix},\\
    \pi_3=\begin{pmatrix}
    1&2&3&4&5&6&7&8&9\\
    9&8&7&6&5&4&3&2&1
    \end{pmatrix}, \pi_4=\begin{pmatrix}
    1&2&3&4&5&6&7&8&9\\
    7&4&1&8&5&2&9&6&3
    \end{pmatrix}.
    \end{split}
\end{equation*}
Let $\Pi_{\text{rotation}}=\{\pi_1,\pi_2,\pi_3,\pi_4\}$. By applying the permutations in $\Pi_{\text{rotation}}$ to $\hat{R}$, we could simulate any growing reactions $(A,\Ocal,B,C,d)$ toward different directions.

Let $\Pi_k=\{\pi'\circ\pi_k:\pi'\in\pi_{\text{transition}}\}$, $k\in[4]$. By applying $\pi\in\Pi_k$ to $\hat{R}$, we can simulate any reactions starting from $A^i,i\in[9]$ toward a specific direction $d$ relative to $k$. The entire simulation protocol for a growing reaction $(A,\Ocal,C,D,d)$ is described below: For all growing reactions $(A,\Ocal,B,C,d)\in R_g$, pick $k\in[4]$ s.t. $(d,k)\in\{(\rw,1),(\uw,2),(\lw,3),(\dw,4)\}$. For all $\pi\in\Pi_k$, add the following reactions to $R'$:
\begin{enumerate}
    \item $A^{\pi(1)}+\Ocal^{\pi(2)}\rw A^{\pi(1)}+\chi_{\pi(1)}^{\pi(2)}$.\comm{Triggering the local coloring process.}
    
    \item $\chi_{\pi(1)}^{\pi(2)}+\psi^{\pi(3)}\rw\chi^{\pi(2)}+\psi^{\pi(3)}$, for all $\psi\in\dagQ\cup\{\chi\}$.
    
    \comm{$\Bcal(\chi_{\pi(1)}^{\pi(2)})$ has been colored.}
    
    \item $\chi_{\pi(1)}^{\pi(2)}+\Ocal/\Ocal^{\pi(3)}\rw\chi_{\pi(1)}^{\pi(2)}p+\chi_{\pi(2)}^{\pi(3)}00$.\\
    $\chi_{\pi(1)}^{\pi(2)}+\chi_{\pi(6)}^{\pi(3)}/\chi_{\pi(9)}^{\pi(3)}\rw\chi_{\pi(1)}^{\pi(2)}p+\chi_{\pi(2)}^{\pi(3)}00$.
    
    \comm{$\Bcal(\chi_{\pi(1)}^{\pi(2)}p)\setminus\Bcal(A^1)$ is going to be colored.}
    
    \item $\chi_{\pi(1)}^{\pi(2)}+\chi_{\pi(1)}^{\pi(3)}\rw\chi^{\pi(2)}+\chi^{\pi(3)}$.

    \comm{Both blocks $\Bcal(\chi_{\pi(1)}^{\pi(2)}),\Bcal(\chi_{\pi(1)}^{\pi(3)})$ have been colored.}
    
    \item $\chi_{\pi(2)}^{\pi(3)}0b+\Ocal\rw\chi_{\pi(2)}^{\pi(3)}0b+\chi_{23}$, for all $b\in\{0,1\}$.\\ 
    $\chi_{\pi(2)\pi(3)}+\psi\rw\Ocal^{\pi(6)}+\psi$, for all $\psi\in\Xi^5$.\\
    $\chi_{\pi(2)}^{\pi(3)}b0+\Ocal\rw\chi_{\pi(2)}^{\pi(3)}b0+\chi_{\pi(2)\pi(3)}, \text{ for all }b\in\{0,1\}$.\\
    $\chi_{\pi(2)\pi(3)}+\psi\rw\Ocal^{\pi(9)}+\psi, \text{ for all }\psi\in\Xi^8$.

    \comm{A coloring species observes the adjacent blank species $\Ocal$.}
    
    \item $\chi_{\pi(2)}^{\pi(3)}0b+\psi\rw\chi_{\pi(2)}^{\pi(3)}1b+\Ocal^{\pi(6)}, \text{ for all }\psi\in\{\chi_{\pi(3)\pi(9)},\chi_{\pi(4)\pi(5)},\chi_{\pi(5)\pi(4)}\}$.\\
    $\chi_{\pi(2)}^{\pi(3)}b0+\psi\rw\chi_{\pi(2)}^{\pi(3)}b1+\Ocal^{\pi(9)}, \text{ for all }\psi\in\{\chi_{\pi(3)\pi(6)},\chi_{\pi(7)\pi(8)},\chi_{\pi(8)\pi(7)}\}$.

    \comm{A coloring species meets the redundant states from other growing reactions}
    
    \item $\chi_{\pi(2)}^{\pi(3)}0b+\chi_{\pi(5)}^{\pi(6)}b0\rw\chi_{\pi(2)}^{\pi(3)}1b+\chi_{\pi(5)}^{\pi(6)}b1, \text{ for }b=0,1$.\\
    $\chi_{\pi(2)}^{\pi(3)}0b+\chi_{\pi(4)}^{\pi(6)}0b\rw\chi_{\pi(2)}^{\pi(3)}1b+\chi_{\pi(4)}^{\pi(6)}1b, \text{ for }b=0,1$.\\
    $\chi_{\pi(2)}^{\pi(3)}b0+\chi_{\pi(8)}^{\pi(9)}0b\rw\chi_{\pi(2)}^{\pi(3)}b1+\chi_{\pi(8)}^{\pi(9)}1b, \text{ for }b=0,1$.\\
    $\chi_{\pi(2)}^{\pi(3)}b0+\chi_{\pi(7)}^{\pi(9)}b0\rw\chi_{\pi(2)}^{\pi(3)}b1+\chi_{\pi(7)}^{\pi(9)}b1, \text{ for }b=0,1$.

    \comm{Two coloring species observing each other's.}
    
    \item 
    $\chi_{\pi(2)}^{\pi(3)}0b+\psi\rw\chi_{\pi(2)}^{\pi(3)}1b+\psi, \text{ for all }\psi(v_N)\in\Xi^{\pi(6)}\setminus\{\chi_{\pi(5)}^{\pi(6)}b0,\chi_{\pi(4)}^{\pi(6)}0b\}$.\\
    $\chi_{\pi(2)}^{\pi(3)}b0+\psi\rw\chi_{\pi(2)}^{\pi(3)}b1+\psi, \text{ for all }\psi(v_S)\in\Xi^{\pi(9)}\setminus\{\chi_{\pi(8)}^{\pi(9)}0b,\chi_{\pi(7)}^{\pi(9)}b0\}$.

    \comm{The cell has been colored.}

    \item $\chi_{\pi(1)}^{\pi(2)}p+\chi_{\pi(2)}^{\pi(3)}11\rw\chi^{\pi(2)}+\Ocal^{\pi(3)}$.\comm{Announce the termination of local coloring.}

    \item $\Ocal^{\pi(3)}+\chi_{\pi(2)\pi(3)}\rw\Ocal^{\pi(3)}+\Ocal$.
    
    \comm{Clear off the redundant states produced by the coloring species.}

    \item $A^{\pi(1)}+\chi^{\pi(2)}\rw B^{\pi(1)}+C^{\pi(2)}$.\comm{Perform the state transition.}
    
\end{enumerate}

As we complete the simulation of growing reactions, we now briefly explain how to simulate the other reactions.

\paragraph{\texttt{State\_Transitions}}
For the reactions happens on the cells whose block have been colored, the simulation is straight forward. For a bimolecular reaction $(A,B,C,D,d)\in R\setminus R_g$, like the growing reactions, we have to include the reactions on all $A^i,i\in[9]$. Pick $k\in[4]$ s.t. $(d,k)\in\{(\rw,1),(\uw,2),(\lw,3),(\dw,4)\}$. For all $\pi\in\Pi_k$, add the following reactions to $R'$:
\begin{enumerate}
    \item $A^{\pi(1)}+B^{\pi(2)}\rw C^{\pi(1)}+D^{\pi(2)}$.
\end{enumerate}
For a unimolecular reaction $(A,B,\odot)$, the simulation is simple: For all $i\in[9]$, add reactions to $R'$:
\begin{enumerate}
    \item $A^i\rw B^i$.
\end{enumerate}

So far we have given the entire simulation for a s-d-sCRN $\Gamma=(Q,S,R)$ using a s-sCRN $\Gamma'=(Q',S',R')$ by combining the above three parts.

\subsection{Proof sketch}\label{ssec:4.proof}
First we give the representation function $\Rcal:Q'\rw Q$. Define $\hat{Q}=\{\sigma^i:\sigma\in Q\}$. Then $\Rcal$ is defined as the following:
\begin{itemize}
    \item For all $\psi=\sigma^i\in\hat{Q}$, $\Rcal(\psi)=\sigma$.
    \item $\Rcal(0)=s$.
    \item For all $\psi\not\in\hat{Q}\cup\{0\}$, $\Rcal(\psi)=\Ocal$.
\end{itemize}
By Definition~\ref{def:sim}, we need to show that
\begin{enumerate}
    \item $\Gamma\flw_\Rcal\Gamma'$.
    \item $\Gamma'\models_\Rcal\Gamma$.
    \item $\Gamma'\Lrw_\Rcal\Gamma$.
\end{enumerate}

\subsubsection{$\Gamma\flw_\Rcal\Gamma'$}
\begin{lemma}\label{lem:1}
    $\alpha'\rw^1\beta'$ for some $\alpha',\beta'\in\Acal(\Gamma')$ $\implies$ $\Rcal^*(\alpha')\rw^1 \Rcal^*(\beta')$.
\end{lemma}
Lemma~\ref{lem:1} implies $\Gamma\flw_\Rcal\Gamma'$. If we have $\alpha'\rw_{\Gamma'}\beta'$ for some $\alpha',\beta'\in\Acal(\Gamma')$, there exists a sequence of configurations $\alpha'=\alpha_0,\alpha_1,\cdots,\alpha_k=\beta'$ s.t. $\alpha'_{i-1}\rw^1\alpha'_i$ for all $i\in[k]$. With this lemma we know $\Rcal^*(\alpha'_{i-1})\rw^1 \Rcal^*(\alpha'_i)$ for every $i\in[k]$. Therefore $\Rcal^*(\alpha')\rw_\Gamma\Rcal^*(\beta')$, which implies that $\Gamma\flw_\Rcal\Gamma'$.

\begin{proof}of Lemma~\ref{lem:1}.

Suppose that $\beta'$ can be produced by applying reactions $r'\in R'$ to $\alpha'$.

If $r'$ is a unimolecular reaction $A^i\rw B^i$ for some $A\neq\Ocal$, then
there must be a reaction $(A,B,\odot)=(\Rcal^*(A^i),\Rcal^*(B^i),\odot)\in R$ by our construction. This implies that $\Rcal^*(\alpha')\rw^1\Rcal^*(\beta')$.

Similarly, if $r'$ is a bimolecular reaction of form $A^{\pi(1)}+B^{\pi(2)}\rw C^{\pi(1)}+D^{\pi(2)}$ for some $A,B\in\dagQ,C,D\in Q$, then there exists a reaction $(A,B,C,D,d)\in R$ where $(A,B,C,D,d)=(\Rcal^*(A^{\pi(1)}),\Rcal^*(B^{\pi(2)}),\Rcal^*(C^{\pi(1)}),\Rcal^*(D^{\pi(2)}),d)$, which implies that $\Rcal^*(\alpha')\rw^1\Rcal^*(\beta')$.

If $r'$ comes from the simulation process of \ptc{Growing\_Reactions}, then there exists a growing reaction $(A,\Ocal,B,C,d)\in R_g$ ($d$ corresponds to the permutation $\pi$ used in $r'$). In this sub-protocol, the only reaction that will result in a state transition in $\Rcal^*(\alpha')$ is $r'=\{A^{\pi(1)}+\chi^{\pi(2)}\rw B^{\pi(1)}+C^{\pi(2)}\}$.\\
Since $(\Rcal^*(A^{\pi(1)}),\Rcal^*(\chi^{\pi(2)}),\Rcal^*(B^{\pi(1)}),\Rcal^*(C^{\pi(2)}),d)=(A,\Ocal,B,C,d)\in R_g$, we have $\Rcal^*(\alpha')\rw^1\Rcal^*(\beta')$. Every other reaction is not going to change the resulting configuration under the mapping $\Rcal^*$, i.e. $\Rcal^*(\alpha')=\Rcal^*(\beta')$.

Else, $r'$ comes from the execution of \ptc{Determine\_Global\_Orientation}. In this case, $\Rcal^*(\alpha')=\Rcal^*(\beta')$ always. 

\end{proof}

\subsubsection{$\Gamma'\models_\Rcal\Gamma$}
For every configuration $\alpha\in\Acal(\Gamma)$, let $\Pi(\alpha)={\Rcal^*}^{-1}(\alpha)\cap\Acal(\Gamma')$, which is the reachable configuration in $\Gamma'$ that mapped to $\alpha$. Then it remains to show the following lemma:
\begin{lemma}\label{lem:2}
    \begin{enumerate}
    \item $\Pi(\alpha)\neq\emptyset$.
    \item Given any $\beta$ s.t. $\alpha\rw^1\beta$, for every $\alpha'\in\Pi(\alpha)$, there exists $\beta'$ s.t. $\Rcal^*(\beta')=\beta$ and $\alpha'\rw\beta'$.
\end{enumerate}
\end{lemma}

This indeed implies the original definition of $\Gamma'\models_\Rcal\Gamma$. Since every $\alpha''$ s.t. $\Rcal^*(\alpha'')=\alpha$, $\alpha''$ is included in $\Pi(\alpha)$. And for any $\beta$ s.t. $\alpha\rw\beta$, we can find a sequence of configurations $\alpha_1,\cdots,\alpha_k$ s.t. $\alpha=\alpha_1\rw^1\alpha_2\rw^1\cdots\rw^1\alpha_k=\beta$. Then by the claim, for every $\alpha'\in\Pi(\alpha)$, there exists $\alpha'_1$ s.t. $\Rcal^*(\alpha'_1)=\alpha_1$ and $\alpha'\rw\alpha'_1$, which implies that $\alpha'_1\in\Pi(\alpha_1)$. By similar argument, we could show that there exists $\alpha'_i$ s.t. $\Rcal^*(\alpha'_i)=\alpha_i$ and $\alpha'_{i-1}\rw\alpha'_i$ for every $i=2,\cdots,k$. Hence, there exists $\beta'$ s.t. $\Rcal^*(\beta')=\beta$ and $\alpha'\rw\beta'$.

To prove the lemma we first give another lemma.
\begin{lemma}\label{lem:3}
Let $\Acal(s^5,\Gamma')$ be the set of configurations reachable from $\{\alpha':s^5\in\alpha'\text{ and }\alpha'\in\Pi(S)\}$. i.e. the reachable configurations starting after the global orientation has been determined by coloring the block of $s^5$. Then for every $\alpha'\in\Acal(s^5,\Gamma')$, $\alpha'$ satisfies the complete coloring property.
\end{lemma}

\begin{proof}of Lemma~\ref{lem:3}.

We prove by saying that for every $\alpha'\in\Acal(s^5,\Gamma')$, if $\alpha'$ satisfies the complete coloring property, then for any $\beta'$ s.t. $\alpha'\rw^1\beta'$, $\beta'$ satisfies the complete coloring property as well.

Suppose that $\Rcal^*(\alpha')=\alpha$, $\Rcal^*(\beta')=\beta$, and that $\beta'$ is produced by $\alpha'$ applying a reaction $r'\in R'$. Define reaction set $R'_g=\{A^{\pi(1)}+\chi^{\pi(2)}\rw C^{\pi(1)}+D^{\pi(2)}:A,B,C\in\dagQ,\pi\in\cup_{k=1}^4\Pi_k\}$. Except for the case $r'\in R'_g$, there will be no species in $\dagQ$ appears in $\beta$, and there is no reaction that eliminate the existing coloring throughout the whole simulation, so the complete coloring property is preserved.

When $r'\in R'_g$, it suffices to show that each time a species $\chi^i$ appears, $\Bcal(\chi^i)$ has always been colored (This is the property we used in \ptc{Growing\_Reactions}). First observe that $\chi^i$ could be produced only by item $2.$, $4.$ and $9.$ in the simulation of growing reactions. For simplicity, we fix $\pi=id$ in the succeeding description. In item $2.$, it is clear that $\Bcal(\chi^2)$ has been colored. In item $4.$, we know that $\chi_1^2$ is triggered from some species $A^1,A\in\alpha'\cap\dagQ$ and $\chi_1^3$ is triggered from some species $B^1,B\in\alpha'\cap\dagQ$. By the assumption that $\alpha'$ is completely colored, it is guaranteed that $\Bcal(\chi_1^2)$ and $\Bcal(\chi_1^3)$ have both been colored. Therefore we could turn them into $\chi^2,\chi^3$ simultaneously.

For the case in item $9.$, $\chi_2^300$ must have been produced first (item $3.$). To understand what happened in between, we have to take a closer look at the behavior of the following $3$ kinds of species.

$\bs{\chi_{23}}$ - $\chi_{23}$ can only be produced by $\chi_2^3b_0b_1$ (for some $b_0,b_1$ not all $1$) acting on some neighboring blank state $\Ocal$, and from $\chi_2^3b_0b_1$ we can trace back to $\chi_2^300$, which is triggered by $\chi_1^2$ (item $3.$). $\chi_1^2$ must be a neighbor of some $A^1\in\dagQ$. By the assumption that $\alpha'$ is completely colored, $\Bcal(A^1)$ is colored, and hence any $\chi_{23}$ in $\Bcal(\chi_1^2)\setminus\Bcal(A^1)$ will encounter some species in $\Xi^5$ or $\Xi^8$ and eventually be turned into $\Ocal^6,\Ocal^9$ respectively by item $5.$. Note that there may be one redundant $\chi_{23}$ appears outside $\Bcal(\chi_1^2p)$, it should be turned back into $\Ocal$ later.

$\bs{\chi_2^3b_0b_1\textbf{ and }\chi_1^2p}$ - We denote by $W$ to represent $\Bcal(\chi_1^2)\setminus\Bcal(A^1)$. By item $5.$, $\chi_2^3b_0b_1$ turns $\Ocal\in W$ into $\chi_{23}$, which will eventually become $\Ocal^6$ or $\Ocal^9$. If meeting species in $\{\chi_{\pi(3)\pi(9)},\chi_{\pi(4)\pi(5)},\chi_{\pi(5)\pi(4)}\}$ or $\{\chi_{\pi(3)\pi(6)},\chi_{\pi(7)\pi(8)},\chi_{\pi(8)\pi(7)}\}$, $\chi_2^3b_0b_1$ knows which color the position must be, so it turned them into $\Ocal^6$ or $\Ocal^9$ directly. For the remaining $\chi_{ij}$ it may encounter, they are a part of some other wall being colored, so they eventually turns into $\Ocal^6,\Ocal^9$ as well. Thus, $\chi_2^300$ eventually sees two colored species on $W$, one in $\Xi^6$ and one in $\Xi^9$. Item $7.$ gives the corresponding reactions for each situation, which always turns $\chi_2^300$ into $\chi_2^311$ eventually. $\chi_2^311$ is served as a signal that announce the completion of coloring $W$. Notice that in item $3.$, after $\chi_1^2$ trigger the formation of $\chi_2^300$, it becomes $\chi_1^2p$ and then stay still waiting for $\chi_2^311$ to turn it into $\chi^2$ (item $9.$). In this case, $\Bcal(\chi^2)$ is indeed colored. On the other hand, the only reaction $\chi_2^311$ can perform is also item $9.$, hence $\chi_2^211$ will eventually become $\Ocal^3$, and then turn the redundant $\chi_{23}$ back to $\Ocal$. The proof complete.

\end{proof}

Now we are ready to prove Lemma~\ref{lem:2}.

\begin{proof}of Lemma~\ref{lem:2}.

Initially, $\Pi(S)\neq\emptyset$ and $\Pi(S)$ satisfies the complete coloring property for sure. By the construction of \ptc{Determine\_Global\_Orientation}, for every $\alpha'\in\Pi(S)\setminus\Acal(s^5,\Gamma')$ there exists $\beta'\in\Pi(S)\cap\Acal(s^5,\Gamma')$ reachable for $\alpha'$. For any $\alpha\in\Acal(\Gamma)$, $\alpha\neq S$, $\Pi(\alpha)={\Rcal^*}^{-1}(\alpha)\cap\Acal(s^5,\Gamma')$ since each of them can only be produced after $s^5$ appeared. By Lemma~\ref{lem:1}, for any $\alpha\in\Acal(\Gamma)$, $\Pi(\alpha)\cap\Acal(s^5,\Gamma')$ satisfies the complete coloring property. And we are going to show that for any $\alpha'\in\Pi(\alpha)\cap\Acal(s^5,\Gamma')$, there is always a sequence of reactions that takes $\alpha'$ to some $\beta'$ s.t. $\Rcal^*(\beta')=\beta$. With this property, we have that for every $\alpha\in\Acal(\Gamma)$ there exists a sequence of reactions which takes $s'$ to $\alpha'\in\Pi(\alpha)$ that pass through some $\alpha''\in\Pi(S)\cap\Acal(s^5,\Gamma')$. Thus $\Pi(\alpha)\neq\emptyset$.  

Suppose that $\beta$ is produced by $\alpha$ applying a reaction $r\in R$. If $r\not\in R_g$, by the fact that $\alpha'$ is completely colored, we can apply the reactions correspond to $r$ in \ptc{State\_Transitions} to obtain a configuration $\beta'$ s.t. $\Rcal^*(\beta')=\beta$.

When $r\in R_g$, w.l.o.g. we assume that $r=(A,\Ocal,B,C)$ acting on cell $v$ which is colored by $1$ in $\alpha'$. And in $\alpha'$ the cell $v'$ to the East of $v$ is now containing a state $\sigma(v')$ colored $2$ that mapped to $\Ocal$. Notice that it does not matter how fast the other cells turn into state $\chi^i$. Roughly speaking, we could postpone the legal state transitions on those cells from then on, rather than performing them immediately. So it suffices to show that, at this point, there exists a sequence of reactions that turns $v'$ to state $\chi^2$.  Recall that there are $8$ kinds of states $\sigma'=\sigma(v')$ could be, and we'll discuss them all.

$\bs{\sigma'=\sigma'^2, \chi^2}$ - Then it is done.

$\bs{\sigma'=\Ocal^2}$ - $\sigma(v')$ will be triggered by $A^1$ to become $\chi_1^2$.

$\bs{\sigma'=\chi_i^jb_0b_1,\chi_{ij}}$ - By the proof of Lemma~\ref{lem:1}, $\sigma'$ will become $\Ocal$  (impossible for $v'$ since it has been colored) or $\Ocal^i$. Go to the above case.

$\bs{\sigma'=\chi_i^2p}$ - By the proof of Lemma~\ref{lem:1}, $\sigma'$ will eventually become $\chi^j$.

Now it suffice to discuss the case $\sigma'=\chi_i^2$, w.l.o.g. we let $i=1$. Since $\Bcal(A^1)$ has been colored and there is no reaction between $\chi_1^2$ and $\Xi^5$ or $\Xi^8$, so we only need to analyze the state $\psi$ to the East of $v'$.

$\bs{\psi=\Ocal,\Ocal^3}$- Then $v'$ turns into $\chi_1^2p$. By the proof of Lemma~\ref{lem:1}, $v'$ will become $\chi^2$ with its block colored.

$\bs{\psi}=\chi_i^jb_0b_1,\chi_{ij}$ - By the proof of Lemma~\ref{lem:1}, $\psi$ will become $\Ocal$ or $\Ocal^i$. Go to the above case.

$\bs{\psi=\sigma^3}$ - $\Bcal(\chi_1^2)$ has been colored, turns itself into $\chi^2$.

$\bs{\psi=\chi^3}$ - Turns itself into $\chi^2$. We'll explain this later.

$\bs{psi=\chi_i^3p}$ - By the proof of Lemma~\ref{lem:1}, $\psi$ will eventually become $\chi^3$, go to the above case.

$\bs{\psi=\chi_1^3,\chi_9^3,\chi_6^3}$ - If $\psi=\chi_1^3$, then both the block of $\chi_1^2,\chi_1^3$ have been colored by the assumption of complete coloring of $\alpha'$. If $\psi=\chi_9^3,\chi_6^3$, perform item $3.$ and then go to the first case.

Now, except for the case $\psi=\chi^3$, each time $\chi^2$ appears, $\Bcal(\chi^2)$ has been colored. If it is the existing $\chi^3$ that cause the formation of $\chi^2$, we can ask what cause the formation of $\chi^3$. So we'll have a sequence $(\chi^{j(k)})_{k=0}^n$ where $j(n)=2,j(n-1)=3, \cdots$ s.t. the existence of $\chi^{j(k-1)}$ caused the formation of $\chi^{j(k)}$ for all $j=1,\cdots,n$, and $\Bcal(\chi^{j(k)})$ is colored as long as $\Bcal(\chi^{j(k-1)})$ is colored. Since it is the unit-seeded model, the first $\chi^{j(0)}$ exists, which must be produced only if $\Bcal{\chi^{j(0)}}$ is colored. Thus $\Bcal(\chi^{j(n)})$ is colored. The proof complete.
    
\end{proof}


\subsubsection{$\Gamma'\Lrw_\Rcal\Gamma$}
First we prove that $\{\Rcal^*(\alpha'|\alpha'\in\Acal(S))\}=\Acal(T)$.
\begin{lemma}\label{lem:4}
    \begin{enumerate}
    \item $\forall \alpha\in\Acal(\Gamma), \exists\alpha'\in\Acal(\Gamma')\text{ s.t. }\Rcal^*(\alpha')=\alpha$.
    \item $\forall \alpha'\in\Acal(\Gamma'), \Rcal^*(\alpha')\in\Acal(\Gamma)$.
\end{enumerate}
\end{lemma}
The first item implies that $\{\Rcal^*(\alpha'|\alpha'\in\Acal(S))\}	\supseteq\Acal(T)$, and the second one implies $\{\Rcal^*(\alpha'|\alpha'\in\Acal(S))\}\subseteq\Acal(T)$.
\begin{proof}
\begin{enumerate}
    \item Given $\alpha$, let $S=\alpha_0\rw^1\alpha_1\rw^1\cdots\rw^1\alpha_n=\alpha$ be any sequence let achieve the configuration $\alpha$, and assume that $\alpha_i$ is produced by applying reaction $r_i$ to $\alpha_{i-1}$ for all $i\in[n]$. Then we construct a sequence of reaction to produce $\alpha'$. First we use \texttt{Determine\_Global\_Orientation} to get $\alpha_0'$. And each $r_i$ corresponds to a sequence of reactions $r'_{i_1},\cdots,r'_{i_m(i)}$ in the simulation. It is obvious that by applying these reaction sequences $(r'_{i_k})_{k=1}^{m(i)}$ for $i=1,\cdots,n$, we'll get the resulting configuration $\alpha'$ s.t. $\Rcal^*(\alpha')=\alpha$.
    \item The reactions in $R'$ that affect the image of $\Rcal^*$ are those described in \texttt{State\_Transitions} and item $11.$ of \texttt{Growing\_Reactions}. And each of them correspond to a reaction in $R$. Therefore, $\Rcal^*(\alpha')$ is achievable in $\Gamma$ as well.
\end{enumerate}
\end{proof}
Then it remains to show the following lemma:
\begin{lemma}\label{lem:5}
    \begin{enumerate}
    \item Given $\alpha\in\Acal(\Gamma)\setminus\Acal_*(\Gamma)$, then $\forall\alpha'\text{ s.t. }\Rcal^*(\alpha')=\alpha$, $\alpha'\not\in\Acal_*(\Gamma)$.
    \item $\forall \alpha\in\Acal_*(\Gamma), \exists\alpha'\in\Acal(\Gamma')\text{ s.t. }\Rcal^*(\alpha')=\alpha$.
\end{enumerate}
\end{lemma}

\begin{proof}
\begin{enumerate}
    \item It means that there exists $r\in R$ s.t. applying $r$ to $\alpha$ produces some $\beta\neq\alpha$. So the statement follows by the analysis in the proof of $\Gamma'\models\Gamma$.
    \item By the proof of Lemma~\ref{lem:4} item $1$, there exists a sequence of reactions $r'$ that produces $\alpha'$. And every reaction sequence that leads to $\beta'$ s.t. $\Rcal^*(\beta')\neq\Rcal^*(\alpha')$ corresponds to some reactions in $R$ by the proof of Lemma~\ref{lem:4} item $2$.
\end{enumerate}
\end{proof}

%% file: sCRN-aTAM_v8.tex
\section{Simulate aTAM by unit-seeded directed sCRN}\label{sec:aTAM}
In this section, we demonstrate that unit-seeded directed sCRN can simulate aTAM.
\begin{theorem}\label{thm:dsCRN-aTAM}
Given a system of aTAM $\Gamma=(Q,S,g,\tau)$, there exists a unit-seeded directed sCRN $\Gamma'=(Q',R)$ that simulates $\Gamma$.
\end{theorem}

\subsection{Simulation overview}
The primary task is to simulate the attachment of a single tile. Essentially, we encode each tile $t$ by its state $\sigma(t)=(\sigma_N(t),\sigma_E(t),\sigma_S(t),\sigma_W(t))$. The main challenge arises from the fact that whether a tile is attachable in a position depends on all $4$ neighbors. However, in a d-sCRN, the transition of a cell's state is determined by only one of its neighbors. To address this, we introduce some auxiliary variables to enable a cell to contain a special species capable of "observing" all species around. This process is described in \ptc{Observe\_Neighbors}. Subsequently, it determines whether there exists any attachable tile in that cell. This aspect is elaborated in \ptc{Attachment}.

More precisely, we define two kinds of species: \emph{tile species} and \emph{observing species}, depicted in Figures~\ref{fig:tile_species} and~\ref{fig:observing_species}. The set of tile species is of the form $\psi=(\psi_N,\psi_E,\psi_S,\psi_W)$ where $\psi$ corresponds to a tile state in $Q$. We denote by $\psi(v)$ the tile species located at cell $v$. For all tiles $\sigma\in Q$, we ensure all corresponding tile species $\psi=\sigma$ are contained in $Q'$. An observing species is denoted by $\eta=(\xi,\eta_N,\eta_E,\eta_S,\eta_W)$, where $\xi$ distinguishes it from the tile species. Again, we use $\eta(v)$ to represent an observing species located at cell $v$. Any observing species must be adjacent to at least one of the tile species, and $\eta_d(v)$ records the labels facing $v$ from the $d$ direction. That is, let $v=(x_0,y_0)$, then $\eta(v)=(\psi_S(x_0,y_0+1),\psi_W(x_0+1,y_0),\psi_N(x_0,y_0-1),\psi_E(x_0-1,y_0))$. Let $\Sigma=\{\iota: \iota\text{ is a label of some }\sigma\in Q\}$ ($\nullsf\in\Sigma$), and let $\Sigma_\varepsilon=\Sigma\cup\{\varepsilon\}$, then $\eta_\Dcal(v)\in\Sigma_\varepsilon^4$. We use $\varepsilon$ to represent that a side has not been observed yet.

\begin{figure}[htbp]
    \centering
    \begin{subfigure}{0.4\textwidth}
        \centering
        \includegraphics[width=0.3\textwidth]{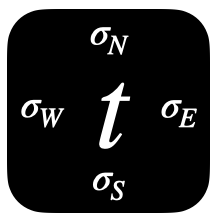}
        \caption{Tile species representing tile $t$.}
        \medskip
        \small $\psi(v)=\sigma(t)=(\sigma_N(t),\sigma_E(t),\sigma_S(t),\sigma_W(t))$.
        \label{fig:tile_species}
    \end{subfigure}%
    \begin{subfigure}{0.4\textwidth}
        \centering
        \includegraphics[width=0.3\textwidth]{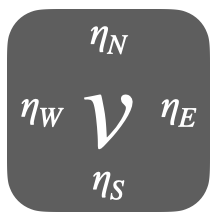}
        \caption{Observing species.}
        \medskip
        \small $\eta(v)=(\eta_N(v),\eta_E(v),\eta_S(v),\eta_W(v))$.
        \label{fig:observing_species}
    \end{subfigure}%
    \caption{Tile species and observing species.}
    \label{fig:tile species and observing species}
\end{figure}

In the simulation, any tile species $\psi$ could change its neighboring blank species $\Ocal$ to an observing species $(\xi,\varepsilon_\Dcal)$. That is, $(\psi,\Ocal,\psi,(\xi,\varepsilon_\Dcal),d')\in R'$ for all $d'\in\{\uw,\rw,\dw,\lw\}$. The observing species then start to record the labels facing itself in each direction. As soon as it has observed each side at least once, it checks whether there exists a legal tile species that is attachable at cell $v$. The entire simulation consists of two parts: \ptc{Observe\_Neighbors} and \ptc{Attachment} in Section~\ref{ssec:5.ptc}. The correctness proof is given in Section~\ref{ssec:5.proof}.

\subsection{Protocols}\label{ssec:5.ptc}

\paragraph{\ptc{Observe\_Neighbors}}
This protocol let an observing species record each label of its neighbor. If an observing species $\eta$ has not observed its neighbor $v$ in direction $d$, then $\eta=(\xi,\lrangle{\varepsilon,\eta_{-d}})$ for some $\eta_{-d}\in\Sigma_\varepsilon^3$. Let $\psi(v)$ be the state of $v$, then $\psi(v)$ may be the following three kinds of species:
\begin{enumerate}
    \item $\psi(v)=\Ocal$. Then $\eta$ records $\nullsf$ in that direction. Therefore, for any\\
    $(d,d')\in\{(N,\uw),(E,\rw),(S,\dw),(W,\lw)\}$, add the following reactions to $R'$:
    \begin{flalign*}
        ((\xi,\lrangle{\varepsilon,\eta_{-d}}),\Ocal,(\xi,\lrangle{\nullsf,\eta_{-d}}),\Ocal,d').&&
    \end{flalign*}
    
    \item $\psi(v)=(\xi,\lrangle{\varepsilon,\eta'_{-d^{-1}}})$ is also an observing species. Then they view each other as a blank species. For any $(d,d')\in\{(N,\uw),(E,\rw),(S,\dw),(W,\lw)\}$, add the following reactions to $R'$:
    \begin{flalign*}
        ((\xi,\lrangle{\varepsilon,\eta_{-d}}),(\xi,\lrangle{\varepsilon,\eta'_{-d^{-1}}}),(\xi,\lrangle{\nullsf,\eta_{-d}}),(\xi,\lrangle{\nullsf,\eta'_{-d^{-1}}}),d').&&
    \end{flalign*}
    
    \item $\psi(v)=\psi_\Dcal$ is a tile species. Then $\eta$ just record the label it sees. For any $(d,d')\in\{(N,\uw),(E,\rw),(S,\dw),(W,\lw)\}$, add the following reactions to $R'$:
    \begin{flalign*}
        ((\xi,\lrangle{\varepsilon,\eta_{-d}}),\psi,(\xi,\lrangle{\psi_{d^{-1}},\eta_{-d}}),\psi,d').&&
    \end{flalign*}
    
\end{enumerate}

Otherwise, the observing species $\eta$ may observe a tile species at $v$ after it has observed a blank species or an observing species at that cell. In other words, we should allow $\eta$ to update its information about $v$ even when it has recorded it to be $\nullsf$. So for any $(d,d')\in\{(N,\uw),(E,\rw),(S,\dw),(W,\lw)\}$, we need the following reactions in $R'$:
\begin{flalign*}
    ((\xi,\lrangle{\nullsf,\eta_{-d}}),\psi,(\xi,\lrangle{\psi_{d^{-1}},\eta_{-d}}),\psi,d').&&
\end{flalign*}

After the observing species has observed all its neighbors, it perform some ``calculation'' to see if there is any legal tile species can attach at its position. The process is described as follows:

\paragraph{\ptc{Attachment}}
In this protocol, an observing species can check whether a legal attachment exists. An observing species $\eta$ has recorded all labels it sees in $\eta_\Dcal\in\Sigma^4$, for any tile species $\psi=\psi_\Dcal\in Q$ satisfies $g(\psi_N,\eta_N)+g(\psi_E,\eta_E)+g(\psi_S,\eta_S)+g(\psi_W,\eta_W)\geq\tau$, we have the following reactions in $R'$:
\begin{flalign*}
    ((\xi,\eta_N,\eta_E,\eta_S,\eta_W),(\psi_N,\psi_E,\psi_S,\psi_W),\odot).&&
\end{flalign*}

Now we summarize the overall simulation protocol.

\paragraph{\ptc{s-d-sCRN\_$\gg$\_aTAM}}
For all $d'\in\{\uw,\rw,\dw,\lw\}$, pick $d$ correspond to $d'$ s.t. $(d,d')\in\{(N,\uw),(E,\rw),(S,\dw),(W,\lw)\}$. Then for all $\eta_{-d},\eta'_{-d^{-1}}\in\Sigma_\varepsilon^3$, $\psi=\psi_\Dcal$ being a tile species, and $d'\in\{\uw,\rw,\dw,\lw\}$, add the following reactions to $R'$ and all occurring species in $Q'$:
\begin{enumerate}
    \item $(\psi,\Ocal,\psi,(\xi,\varepsilon_\Dcal),d')$.
    
    \comm{A tile species turn its neighboring blank species into an observing species.}
    
    \item $((\xi,\lrangle{\varepsilon,\eta_{-d}}),\Ocal,(\xi,\lrangle{\nullsf,\eta_{-d}}),\Ocal,d')$.\comm{An observing species meets $\Ocal$.}
    
    \item $((\xi,\lrangle{\varepsilon,\eta_{-d}}),(\xi,\lrangle{\varepsilon,\eta'_{-d^{-1}}}),(\xi,\lrangle{\nullsf,\eta_{-d}}),(\xi,\lrangle{\nullsf,\eta'_{-d^{-1}}}),d')$.
    
    \comm{Two observing species meet each other.}
    
    \item $((\xi,\lrangle{\varepsilon,\eta_{-d}}),\psi,(\xi,\lrangle{\psi_{d^{-1}},\eta_{-d}}),\psi,d')$.
    
    $((\xi,\lrangle{\nullsf,\eta_{-d}}),\psi,(\xi,\lrangle{\psi_{d^{-1}},\eta_{-d}}),\psi,d').$
    
    \comm{An observing species meets a tile species.}

    \item $((\xi,\eta_N,\eta_E,\eta_S,\eta_W),(\psi_N,\psi_E,\psi_S,\psi_W),\odot)$, 
    
    \hfill if $g(\psi_N,\eta_N)+g(\psi_E,\eta_E)+g(\psi_S,\eta_S)+g(\psi_W,\eta_W)\geq\tau$.\comm{Tile attachment.}

\end{enumerate}
An example of simulating attachment is provided in Figure~\ref{fig:attachment}.

\begin{figure}[htbp]
    \centering
        \includegraphics[width=\textwidth]{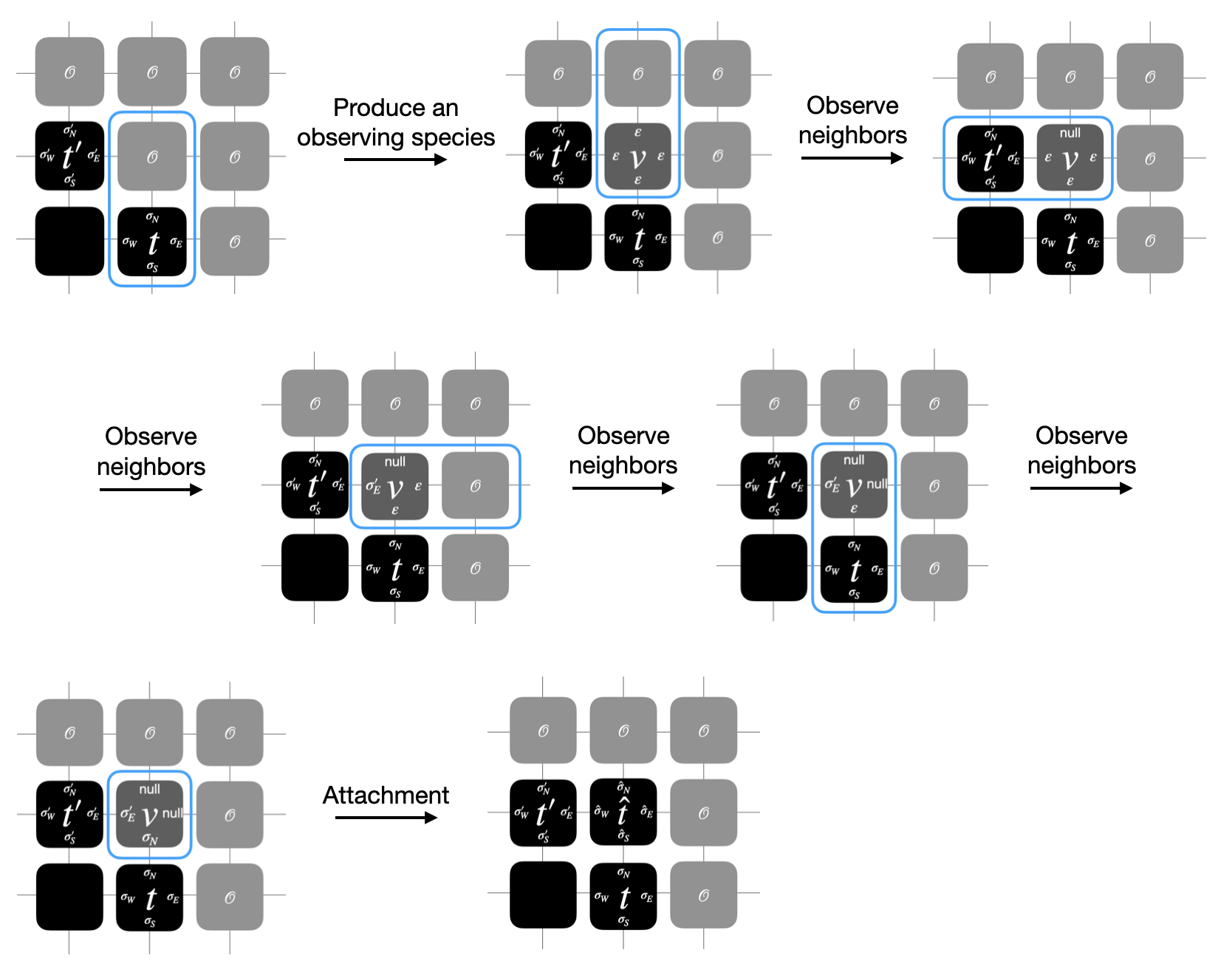}
    \caption{A simple example of simulating attachment.}
    \medskip
    \small Assuming $g(\hat{\sigma}_W,\sigma'_E)=g(\hat{\sigma}_S,\sigma_N)=1$ and $\tau=1$.
    \label{fig:attachment}
\end{figure}

\subsection{Proof sketch}\label{ssec:5.proof}
First, we define the representation function $R:Q'\rightarrow Q$. Let $\hat{Q}=\{\text{tile species}\}$. Then $R$ is defined as follows:
\begin{itemize}
    \item For all tile species $\psi\in\hat{Q}$, $R(\psi)=\psi$.
    \item For all $\psi\not\in\hat{Q}$, $R(\psi)=\nullsf$.
\end{itemize}

\paragraph{$\Gamma\flw_R\Gamma'$}
Observe that in the entire protocol, reactions that produce some tile species $\psi$ that maps to $\psi\neq\nullsf$ come from the \texttt{Attachment} protocol, and the tile species appears only if there exist some tiles attachable to that position. Therefore, for any $\alpha'\rw^1\beta'$ where $R^*(\alpha')=\alpha$ and $R^*(\alpha')=\alpha$, either $\beta=\alpha$ or $\beta=\alpha+t_{(x_0,y_0)}$ for some tile $t$ attachable at position $(x_0,y_0)$ in the aTAM system $\Gamma$.

\paragraph{$\Gamma'\models_R\Gamma$}
For every $\alpha,\beta\in\mathcal{A}(\Gamma)$ such that $\alpha\rw^1\beta$, $\beta=\alpha+t_{(x_0,y_0)}$ for some tile $t$. Assume that $\alpha\in\mathcal{A}(\Gamma')$, then we can always produce $\beta$ by putting an observing species in $(x_0,y_0)$. After observing every existing neighbor, the unimolecular reaction that performs the state transition to a proper tile species follows. Therefore, from a seed $S'=S$, we can produce any configuration of $\Gamma$ in $\Gamma'$. Let $\Pi(\alpha)=\alpha$ for all $\alpha\in\Gamma$, then it remains to show for every $\alpha'\in\Gamma'$ such that $R^*(\alpha')=\alpha$, $\alpha\rw_{\Gamma'}\alpha'$. Notice that $\alpha'-\alpha$ consists of observing species, holding some information of their neighbors. Since the attached tiles won't fall off, we can create those observing species each by growing the species $(\xi,\varepsilon,\varepsilon,\varepsilon,\varepsilon)$ at a right position and have it collect the information of a proper subset of its neighbors.

\paragraph{$\Gamma\Lrw_R\Gamma'$}
By the above explanation, it is obvious that $\{R^*(\alpha')|\alpha'\in\mathcal{A}(\Gamma')\}=\mathcal{A}(\Gamma)$. And the existence of a tile attachment in $\Gamma$ is equivalent to the existence of a tile species attachment in $\Gamma'$, so we have $\{R^*(\alpha')|\alpha'\in\mathcal{A}_*(\Gamma')\}=\mathcal{A}_*(\Gamma)$.

%% file: sCRN-TA_v8.tex
\section{Simulation between unit-seeded TA with affinity-strengthening rule and unit-seeded directed sCRN}\label{sec:TA}

In this section, we prove that unit-seeded directed sCRN and unit-seeded TA with affinity-strengthening rule can simulate each other. 

\begin{theorem}\label{thm:dsCRN-TA}
Given a unit-seeded tile automata system with affinity strengthening rule $\Gamma=(Q,S,I,g,R,\tau)$, there exists a unit-seeded directed sCRN $\Gamma'=(Q',S',R')$ which simulates $\Gamma$.
\end{theorem}

\subsubsection{Simulation overview}\label{sssec:6.1.ov}
The configuration changes in the unit-seeded TA result from either state transitions or tile attachments. The simulation of state transitions is straightforward; we simply view them as bimolecular reactions. We describe this in Section~\ref{sssec:6.1.ptc}, in the \textbf{\ptc{State\_Transitions}} protocol. The tile attachments are similar to those in aTAM. The main issue to be taken care of is that in the a-sa-TA system, it is allowed that a pair of attached tiles change their states according to the transition rule. So the simulation is similar to the one used in simulating aTAM, except that we let the observing species keep updating its neighbors' information so that it is always possible to record the labels consistent with the current configuration. This is described in \textbf{\ptc{Update\_Neighboring\_Labels}} of Section~\ref{sssec:6.1.ptc}, and the correctness proof is given in Section~\ref{sssec:6.1.proof}.

\subsubsection{Protocols}\label{sssec:6.1.ptc}
\paragraph{\ptc{State\_Transitions.}}
To simulate the transition rules in $R$, we simply write each of them as a bimolecular reaction. Thus, we add the following reactions to $R'$:
\begin{flalign*}
    &(\psi_1,\psi_2,\psi_3,\psi_4,\rw),\text{ for all }(\psi_1,\psi_2,\psi_3,\psi_4,\vdash)\in R.\\
    &(\psi_1,\psi_2,\psi_3,\psi_4,\dw),\text{ for all }(\psi_1,\psi_2,\psi_3,\psi_4,\perp)\in R.&&
\end{flalign*}

\paragraph{\ptc{Update\_Neighboring\_Labels.}}
To handle the tile attachments, we use the same designing idea as protocol \ptc{Observe\_Neighbors} in Section~\ref{sec:aTAM}. This protocol makes the observing species keep updating their neighbors' labels. There may exist some tile species performing state transitions after an observing species has recorded its previous state. We slightly modify the \ptc{Observe\_Neighbors} protocol to make an observing species keep updating its information of neighboring labels on each side, even when all $4$ sides have been observed at least once. Notice that now the \emph{tile species} $\psi$ is just encoding a single state in $Q$. For an observing species $\eta$, for all $(d,d')\in\{(N,\uw),(E,\rw),(S,\dw),(W,\lw)\}$, all tile species $\psi\neq\eta_d$, we additionally add the following reactions to $R'$:
\begin{flalign*}
    ((\xi,\eta_N,\eta_E,\eta_S,\eta_W),\psi,(\xi,\lrangle{\psi,\eta_{-d}}),\psi,d').&&
\end{flalign*}

For the entire simulation, the protocol is described below.

\paragraph{\ptc{s-d-sCRN\_$\gg$\_s-as-TA}}
For all $d'\in\{\uw,\rw,\dw,\lw\}$, pick $d$ corresponding to $d'$ such that $(d,d')\in\{(N,\uw),(E,\rw),(S,\dw),(W,\lw)\}$. Then for all $\psi$, where $\psi_i$, $i\in[4]$, is a tile species, $\eta_{-d}$, $\eta'_{-d^{-1}}\in\Sigma_\varepsilon^3$, and $d'\in\{\uw,\rw,\dw,\lw\}$, add the following reactions to $R'$ and all occurring species in $Q'$:
\begin{enumerate}
    \item $(\psi,\Ocal,\psi,(\xi,\varepsilon_\Dcal),d')$.
    
    \comm{A tile species turns its neighboring blank species into an observing species.}
    \item $((\xi,\lrangle{\varepsilon,\eta_{-d}}),\Ocal,(\xi,\lrangle{\nullsf,\eta_{-d}}),\Ocal,d')$.\\
    $((\xi,\lrangle{\varepsilon,\eta_{-d}}),(\xi,\lrangle{\varepsilon,\eta'_{-d^{-1}}}),(\xi,\lrangle{\nullsf,\eta_{-d}}),(\xi,\lrangle{\nullsf,\eta'_{-d^{-1}}}),d')$.\\
    $((\xi,\lrangle{\varepsilon,\eta_{-d}}),\psi,(\xi,\lrangle{\psi_{d^{-1}},\eta_{-d}}),\psi,d')$.\\
    $((\xi,\lrangle{\nullsf,\eta_{-d}}),\psi,(\xi,\lrangle{\psi_{d^{-1}},\eta_{-d}}),\psi,d')$.

    \comm{Behavior of an observing species same as in the simulation of aTAM.}

    \item $((\xi,\eta_N,\eta_E,\eta_S,\eta_W),\psi,(\xi,\lrangle{\psi,\eta_{-d}}),\psi,d')$. \comm{Updating the neighboring labels.}

    \item $((\xi,\eta_N,\eta_E,\eta_S,\eta_W),(\psi_N,\psi_E,\psi_S,\psi_W),\odot)$, 
    
    \hfill if $g(\psi_N,\eta_N)+g(\psi_E,\eta_E)+g(\psi_S,\eta_S)+g(\psi_W,\eta_W)\geq\tau$.\comm{Tile attachment.}

    \item $(\psi_1,\psi_2,\psi_3,\psi_4,\rw)$, for all $(\psi_1,\psi_2,\psi_3,\psi_4,\vdash)\in R$.\\
    $(\psi_1,\psi_2,\psi_3,\psi_4,\dw)$, for all $(\psi_1,\psi_2,\psi_3,\psi_4,\perp)\in R$.\comm{State transition.}
\end{enumerate}

\subsubsection{Proof sketch}\label{sssec:6.1.proof}
First, we give the representation function $R:Q'\rw Q$. Define $\hat{Q}=\{\text{tile species}\}$. Then $R$ is defined as follows:
- For all tile species $\psi\in\hat{Q}$, $R(\psi)=\psi$.
- For all $\psi\not\in\hat{Q}$, $R(\psi)=\nullsf$.

\paragraph{$\Gamma\flw_R\Gamma'$}
A state transition in $\Gamma'$ corresponds directly to a bimolecular reaction in $\Gamma'$, and the attachments are similar to the proof in the simulation of aTAM. Although it may be the case that when an observing species is about to become a tile species, some of its neighbors have performed state transition so that the information is not consistent with what it recorded, the formation of this tile species still corresponds to a legal attachment in the aTAM system due to the affinity-strengthening constraint. Hence $\alpha'\rw_{\Gamma'}^1\beta'\implies R^*(\alpha')\rw_\Gamma^1R^*(\beta')$.

\paragraph{$\Gamma'\models_R\Gamma$}
This follows from the proof in aTAM, since state transition is simulated directly, any $\alpha\in\Acal(\Gamma)$ is also in $\Acal(\Gamma')$. Let $\Pi(\alpha)=\alpha$, then for every attachment in $\Gamma$, we can wait for the observing species to completely update the current states of its neighbors and then perform a corresponding state transition. The remaining proof is the same as in aTAM.

\paragraph{$\Gamma\Lrw_R\Gamma'$}
By the above explanation, it is obvious that $\{R^*(\alpha')|\alpha'\in\Acal(\Gamma')\}=\Acal(\Gamma)$. And the existence of a tile attachment in $\Gamma$ is equivalent to the existence of a tile species attachment in $\Gamma'$ by allowing the observing species to update the information of its neighbors continually, even if there is already a legal attachment that can be performed. Moreover, the updating process stops as long as the information is consistent with all its $4$ neighbors. So we have $\{R^*(\alpha')|\alpha'\in\Acal_*(\Gamma')\}=\Acal_*(\Gamma)$.



Simulating sCRN with tile automata is quite straightforward. We add a \emph{blank tile} with state $\Ocal$ that is attachable to any other tile. All bimolecular reactions are directly translated into state changes of two adjacent tiles. For unimolecular reactions, we allow the tiles with the corresponding state to perform state changes with any neighboring tiles.

%% file: sCRN-asyncCA_v6.tex
\section{Simulation between non-deterministic async-CA and directed sCRN}\label{sec:CA}
In this section, we show that directed sCRN and non-deterministic async-CA can simulate each other. The simulation of non-deterministic async-CA by directed sCRN is described in Section~\ref{sec:CA-sCRN}, and the simulation of directed sCRN by non-deterministic async-CA is described in Section~\ref{sec:sCRN-CA}. With these simulation results, we conclude that the computational power of directed sCRN is the same as non-deterministic async-CA.

\subsection{Simulate non-deterministic async-CA by directed sCRN}\label{sec:CA-sCRN}
In~\cite{CQW20}, they suggest a method to emulate the synchronous cellular automata given a coloring initially. Since sCRN is intrinsically asynchronous, we take asynchronous cellular automata as a target to compare their computational power.
\begin{theorem}\label{thm:dsCRN-CA}
Given a non-deterministic, asynchronous CA $\Gamma=(Q,\Ncal,f)$, there exists a directed sCRN $\Gamma'=(Q',S',R')$ which simulates $\Gamma$.
\end{theorem}

\subsubsection{Simulation overview}\label{sssec:7.1.ov}
Notice that the local function of cellular automata may depend on the orientation, so it is necessary to 
simulated it by a d-sCRN along with the same initial pattern as in $\Gamma$ rather than an (undirected) sCRN. This could be accomplished by providing a predefined coloring on the surface at the beginning. A species must be able to observe and record each of its neighbors, then perform a state transition according to the local function. Like the simulation of aTAM in Section~\ref{sec:aTAM}, we use $4$ variables to record neighboring states. One thing to be careful about is that when a species $\psi$ is observing its neighbors, we require the observed neighbors to keep themselves unchanged until the state transition of $\psi$ is complete. Otherwise, $\psi$ may observe some illegal combination of states in its neighborhood. To avoid this situation, additional $4$ variables are introduced to a species to indicate the lock/unlock relation with its neighbors. This is realized in protocol \ptc{Observe\_and\_Lock}. As stated in the \ptc{State\_transitions} protocol, once a species has locked all its neighbors, it turns itself into another species that represents a different state in the CA system resulting from the local function. Then the species unlocks all its neighbors so that they are released and are able to perform reactions with other species. This is described in \ptc{Release}.

For more details about the simulation, we use a species of this form $\psi(v)=(\sigma(v),\sigma_\Dcal(v),\kappa_\Dcal(v),b,\nu)$ to encode the information needed for cell $v$. Where $\sigma_\Dcal(v)=(\sigma_N(v),\sigma_E(v),\sigma_S(v),\sigma_W(v))$, $\kappa_\Dcal(v)=(\kappa_N(v),\kappa_E(v),\kappa_S(v),\kappa_W(v))$. $\sigma(v)$ is the state of $v$ in the CA system $\Gamma$, $\sigma_\Dcal(v)$ records the states of $v$'s neighbors where $\varepsilon$ means that $v$ has not observed the neighbor on that side yet. $\kappa_\Dcal(v)\in\{\varepsilon,0,1\}$ represent whether $v$ is locked by its neighbor where $1$ means locking others, $0$ means being locked by others, and $\varepsilon$ means that they aren't locked by each other. $\nu\in\{0,1\}$ represents whether $v$ is a \emph{release species}, in which $v$ can do nothing but unlock all its neighbors. For the purpose of making the reaction process end whenever the cellular automata system $\Gamma$ has reached a fixed point, we want a cell not to lock its neighbors twice if the neighborhood stays unchanged. Therefore, we set $b=0$ in the beginning. If the local function $f$ is applied on a neighborhood and no state change is made on $v$, we turn $v$ into a \emph{pause species} where $b=1$. This implies that the activation of $v$ is "paused" until some of its neighbors change. The pause species cannot lock its neighbors. Therefore, if $\Gamma$ has a reachable fixed point, $\Gamma'$ is guaranteed to terminate as soon as every cell has executed at most one round of \ptc{Observe\_and\_Lock}.

If the cellular automata system $\Gamma$ begins with a configuration $S$, then we set the initial configuration of the sCRN to be $S'$ s.t. $S'(v)=(S(v),\varepsilon_\Dcal,\varepsilon_\Dcal,0,0)$ for all $v\in\Zbb^2$, where $\varepsilon_\Dcal=(\varepsilon,\varepsilon,\varepsilon,\varepsilon)$. This means that it has neither observed any of its neighbors nor been locked by them. The basic idea is to have species in the same neighborhood "observe and lock" one another. Roughly speaking, every species tries to lock its neighbors and records their states. Once it has locked all the neighbors successfully, it performs a state transition according to the local function $f$, and becomes a release species ($\nu=1$) to unlock all its neighbors. Note that a locked species is not allowed to lock others to avoid deadlocks, and we further allow each lock to be canceled at any time for simplicity. So the protocol mainly consists of two parts: \ptc{Observe\_and\_Lock} and \ptc{Release}, stated in Section~\ref{sssec:7.1.ptc}. The correctness proof is given in Section~\ref{sssec:7.1.proof}.

\subsubsection{Protocols}\label{sssec:7.1.ptc}

\paragraph{\ptc{Observe\_and\_Lock}}
This protocol describes the process for a species to record its neighbors' states and make them unable to change state until being unlocked. In particular, let a pair of species $\psi,\psi'$ be $(\sigma,\sigma_\Dcal,\kappa_\Dcal,\nu,b)$ and $(\sigma',\sigma'_\Dcal,\kappa'_\Dcal,\nu',b')$ respectively. If $\psi$, not locked by anyone else, observes an adjacent species $\psi'$ and they are not locked by each other, then we allow $\psi$ to lock $\psi'$ and record its state at the same time. So for any $\sigma,\sigma'\in Q$, $\sigma_\Dcal,\sigma_\Dcal\in(Q\cup\{\varepsilon\})^4$, $\kappa_{-d}\in\{\varepsilon,1\}^3$, $\kappa'_{-d^{-1}}\in\{\varepsilon,0,1\}^3$, $b'\in\{0,1\}$, and $(d,d')\in\{(N,\uw),(E,\rw),(S,\dw),(W,\lw)\}$, we add the following reactions to $R'$:
\begin{multline*}
    ((\sigma,\lrangle{\sigma_d,\sigma_{-d}},\lrangle{\varepsilon,\kappa_{-d}},0,0),(\sigma',\lrangle{\sigma'_{d^{-1}},\sigma'_{-d^{-1}}},\lrangle{\varepsilon,\kappa'_{-d^{-1}}},0,b'),\\
    (\sigma,\lrangle{\sigma',\sigma_{-d}},\lrangle{1,\kappa_{-d}},0,0),(\sigma',\lrangle{\sigma,\sigma'_{-d^{-1}}},\lrangle{0,\kappa'_{-d^{-1}}},0,b'),d')
\end{multline*}
At any time we allow $\psi$ to unlock $\psi'$ and make it possible to lock other species. For any $\sigma,\sigma'\in Q$, $\sigma_\Dcal,\sigma'_\Dcal\in(Q\cup\{\varepsilon\})^4$, $\kappa_{-d},\kappa'_{-d^{-1}}\in\{\varepsilon,0,1\}^3$, $b'\in\{0,1\}$, and $(d,d')\in\{(N,\uw),(E,\rw),(S,\dw),(W,\lw)\}$, we have the following reactions in $R'$:
\begin{multline*}
    ((\sigma,\lrangle{\sigma',\sigma_{-d}},\lrangle{1,\kappa_{-d}},0,0),(\sigma',\lrangle{\sigma,\sigma'_{-d^{-1}}},\lrangle{0,\kappa'_{-d^{-1}}},0,b'),\\
    (\sigma,\lrangle{\sigma',\sigma_{-d}},\lrangle{\varepsilon,\kappa_{-d}},0,0),(\sigma',\lrangle{\sigma,\sigma'_{-d^{-1}}},\lrangle{\varepsilon,\kappa'_{-d^{-1}}},0,b'),d').
\end{multline*}

\paragraph{\ptc{State\_Transitions}}
As a species locked all its neighbors, it performs state transition to simulate the local function. Therefore, for all $\sigma,\sigma'\in Q$ and $\sigma_\Dcal\in Q^4$ such that $\sigma'=f(\sigma,\sigma_N,\sigma_E,\sigma_S,\sigma_W)$, we have the following two cases:
\begin{enumerate}
    \item If $\sigma'\neq\sigma$, we just apply the local function and turn the species into release state. Therefore, we add the following reactions to $R'$:
    \begin{flalign*}
        ((\sigma,\sigma_\Dcal,\varepsilon_\Dcal,0,0),(\sigma',\sigma_\Dcal,\varepsilon_\Dcal,1,0),\odot).&&
    \end{flalign*}
    \item If $\sigma\neq\sigma'$, we don't perform any state transition but we turn the species into the pause state, preventing it from locking the same neighborhood again. So we have the following reactions in $R'$:
    \begin{flalign*}
        ((\sigma,\sigma_\Dcal,\varepsilon_\Dcal,0,0),(\sigma',\sigma_\Dcal,\varepsilon_\Dcal,1,1),\odot).&&
    \end{flalign*}
\end{enumerate}
Notice that whenever the pause species has released all its neighbors and investigate any state change among its neighbors, it updates the information and is no longer paused. Therefore, for all $\sigma,\sigma'\in Q$, $\sigma_\Dcal,\sigma'_\Dcal\in(Q\cup\{\varepsilon\})^4$, $\kappa_{-d},\kappa'_{-d^{-1}}\in(\{\varepsilon,0,1\})^3$, $b'\in\{0,1\}$, $(d,d')\in\{(N,\uw),(E,\rw),(S,\dw),(W,\lw)\}$, and $\sigma'\neq\sigma_d$, add the following reactions to $R'$:
\begin{multline*}
    ((\sigma,\lrangle{\sigma_d,\sigma_{-d}},\lrangle{\varepsilon,\kappa_{-d}},0,1),(\sigma',\lrangle{\sigma'_{d^{-1}},\sigma'_{-d^{-1}}},\lrangle{\varepsilon,\kappa'_{-d^{-1}}},0,0),\\
    (\sigma,\lrangle{\sigma',\sigma_{-d}},\lrangle{\varepsilon,\kappa_{-d}},0,0),(\sigma',\lrangle{\sigma,\sigma'_{-d^{-1}}},\lrangle{\varepsilon,\kappa'_{-d^{-1}}},0,0),d').
\end{multline*}

The remaining part is \ptc{Release}, which we now describe.
\paragraph{\ptc{Release}}
In this protocol, a release species unlocks all its neighbors and leaves the release state. We first simulate the unlocking process. For all $\sigma.\sigma'\in Q$, $\sigma_\Dcal\in Q^4$, $\sigma'_\Dcal\in(Q\cup\{\varepsilon\})^4$, $\kappa_{-d},\kappa'_{-d^{-1}}\in(\{\varepsilon,0,1\})^3$, $b'\in\{0,1\}$, and $(d,d')\in\{(N,\uw),(E,\rw),(S,\dw),(W,\lw)\}$, we have the following reaction belong to $R'$: 
\begin{multline*}
    ((\sigma,\lrangle{\sigma',\sigma_{-d}},\lrangle{1,\kappa_{-d}},1,b),(\sigma',\lrangle{\sigma,\sigma'_{-d^{-1}}},\lrangle{0,\kappa'_{-d^{-1}}},0,b'),\\
    (\sigma,\lrangle{\sigma',\sigma_{-d}},\lrangle{\varepsilon,\kappa_{-d}},1,b),(\sigma',\lrangle{\sigma,\sigma'_{-d^{-1}}},\lrangle{\varepsilon,\kappa'_{-d^{-1}}},0,b'),d')
\end{multline*}
And then it leaves the release state by performing an unimolecular reaction. Hence, for all $\sigma\in Q$, $\sigma_\Dcal\in Q^4$, $b\in\{0,1\}$, add the following reactions to $R'$:
\begin{flalign*}
    ((\sigma,\sigma_\Dcal,\varepsilon_\Dcal,1,b),(\sigma,\sigma_\Dcal,\varepsilon_\Dcal,0,b),\odot).&&
\end{flalign*}

In conclusion, we summarize the simulation in the following protocol. See Figure~\ref{fig:lock_and_change} for a diagram.

\begin{figure}[htbp]
    \centering
    \begin{subfigure}{0.5\textwidth}
        \centering
        \includegraphics[width=\textwidth]{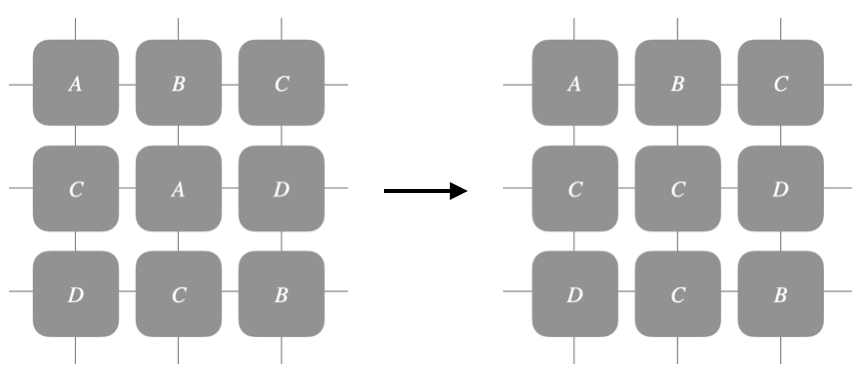}
        \caption{Local function $f(A,B,D,C,C)=C$.}
        \label{fig:local_func}
    \end{subfigure}%
    \begin{subfigure}{0.5\textwidth}
        \centering
        \includegraphics[width=0.6\textwidth]{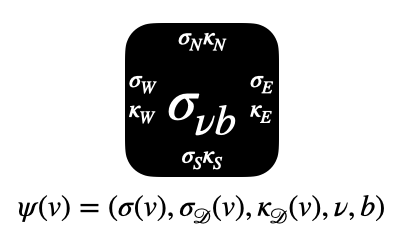}
        \caption{A species encoding information of a cell.}
        \label{fig:CA_species}
    \end{subfigure}%

    \bigskip
    \begin{subfigure}{\textwidth}
        \centering
        \includegraphics[width=\textwidth]{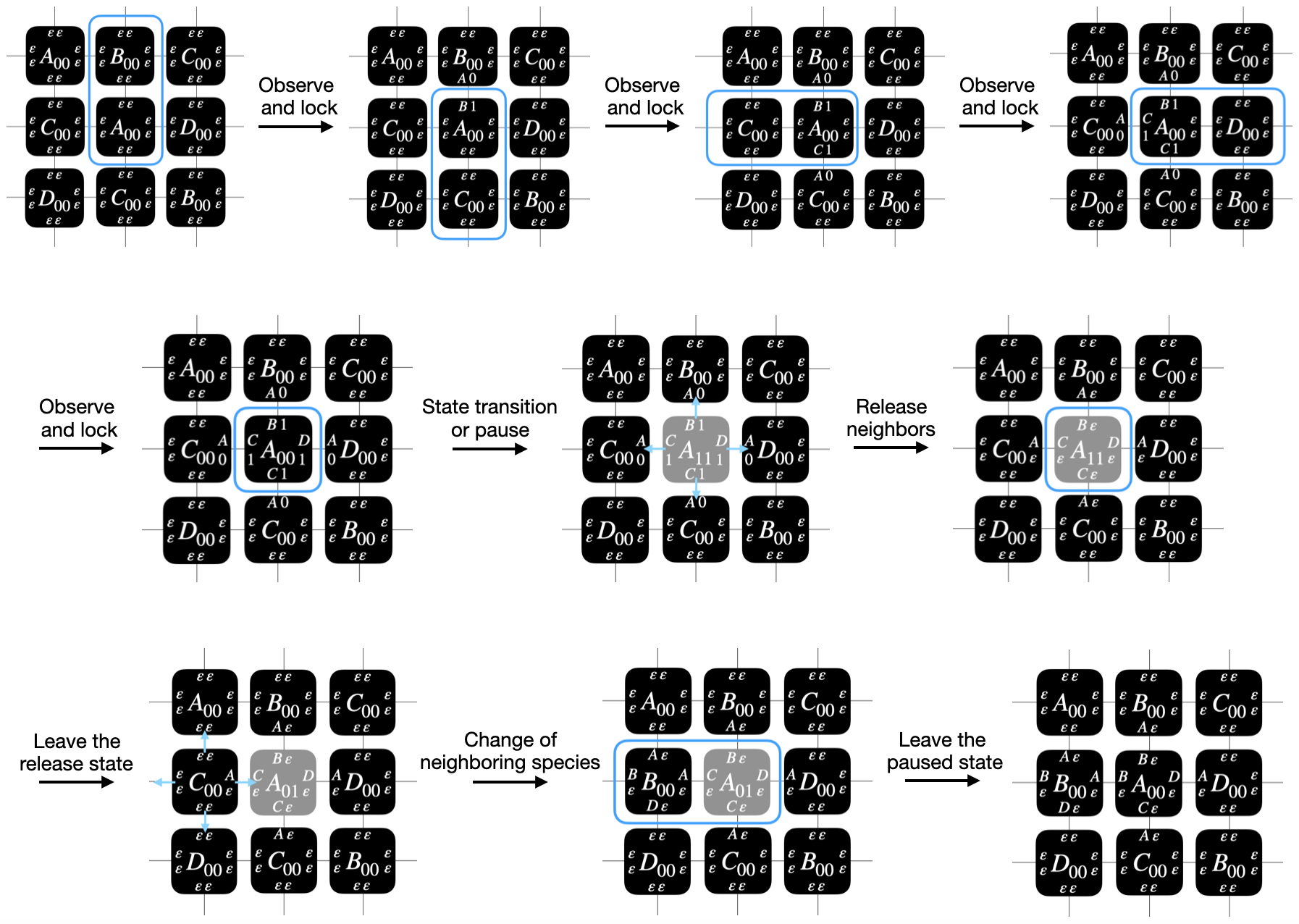}
        \caption{Simulating acync-CA local function $f(A,B,D,C,C)=A$ by d-sCRN.}
        \label{fig:local_function}
    \end{subfigure}%
    \caption{Simulation of acync-CA on sCRN.}
    \label{fig:lock_and_change}
\end{figure}

\paragraph{\ptc{d-sCRN\_$\gg$\_async-CA}}
~\\
For any $\sigma,\sigma'\in Q$, $\sigma_\Dcal,\sigma'_\Dcal\in(Q\cup\{\varepsilon\})^4$, $\kappa_{-d},\kappa'_{-d^{-1}}\in\{\varepsilon,0,1\}^3$, $b'\in\{0,1\}$, and $(d,d')\in\{(N,\uw),(E,\rw),(S,\dw),(W,\lw)\}$, we have the following reactions in $R'$:
\begin{enumerate}
    \item $((\sigma,\lrangle{\sigma_d,\sigma_{-d}},\lrangle{\varepsilon,\kappa_{-d}},0,0),(\sigma',\lrangle{\sigma'_{d^{-1}},\sigma'_{-d^{-1}}},\lrangle{\varepsilon,\kappa'_{-d^{-1}}},0,b'),$
    
    \hfill $(\sigma,\lrangle{\sigma',\sigma_{-d}},\lrangle{1,\kappa_{-d}},0,0),(\sigma',\lrangle{\sigma,\sigma'_{-d^{-1}}},\lrangle{0,\kappa'_{-d^{-1}}},0,b'),d')$,
    
    \hfill if $\kappa_{-d}\in\{1,\varepsilon\}^3$.
    
    \comm{Lock and observe its neighbors.}
    
    \item $((\sigma,\lrangle{\sigma',\sigma_{-d}},\lrangle{1,\kappa_{-d}},0,0),(\sigma',\lrangle{\sigma,\sigma'_{-d^{-1}}},\lrangle{0,\kappa'_{-d^{-1}}},0,b'),$
    
    \hfill $(\sigma,\lrangle{\sigma',\sigma_{-d}},\lrangle{\varepsilon,\kappa_{-d}},0,0),(\sigma',\lrangle{\sigma,\sigma'_{-d^{-1}}},\lrangle{\varepsilon,\kappa'_{-d^{-1}}},0,b'),d')$.
    
    \comm{Unlock its neighbor whenever it want.}
    
    \item $((\sigma,\sigma_\Dcal,\varepsilon_\Dcal,0,0),(\sigma',\sigma_\Dcal,\varepsilon_\Dcal,1,0),\odot)$, if $\sigma'=f(\sigma,\sigma_N,\sigma_E,\sigma_S,\sigma_W)$ and $\sigma\neq\sigma'$.

    \comm{Perform state transition and become release species.}

    \item $((\sigma,\sigma_\Dcal,\varepsilon_\Dcal,0,0),(\sigma',\sigma_\Dcal,\varepsilon_\Dcal,1,1),\odot)$, if $\sigma'=f(\sigma,\sigma_N,\sigma_E,\sigma_S,\sigma_W)$ and $\sigma=\sigma'$.

    \comm{Local function maps the neighborhood of $\sigma$ to itself.}
    
    \comm{$\sigma$ turns itself into the paused release species.}

    \item $((\sigma,\lrangle{\sigma',\sigma_{-d}},\lrangle{1,\kappa_{-d}},1,b),(\sigma',\lrangle{\sigma,\sigma'_{-d^{-1}}},\lrangle{0,\kappa'_{-d^{-1}}},0,b'),$
    
    \hfill $(\sigma,\lrangle{\sigma',\sigma_{-d}},\lrangle{\varepsilon,\kappa_{-d}},1,b),(\sigma',\lrangle{\sigma,\sigma'_{-d^{-1}}},\lrangle{\varepsilon,\kappa'_{-d^{-1}}},0,b'),d')$, if $\sigma_\Dcal\in Q^4$.

    \comm{Unlock all its neighbors after applying the local function.}

    \item $((\sigma,\sigma_\Dcal,\varepsilon_\Dcal,1,b),(\sigma,\sigma_\Dcal,\varepsilon_\Dcal,0,b),\odot)$, if $\sigma_\Dcal\in Q^4$.\comm{Leave the release state.}

    \item $((\sigma,\lrangle{\sigma_d,\sigma_{-d}},\lrangle{\varepsilon,\kappa_{-d}},0,1),(\sigma',\lrangle{\sigma'_{d^{-1}},\sigma'_{-d^{-1}}},\lrangle{\varepsilon,\kappa'_{-d^{-1}}},0,b'),$
    
    \hfill $(\sigma,\lrangle{\sigma',\sigma_{-d}},\lrangle{\varepsilon,\kappa_{-d}},0,0),(\sigma',\lrangle{\sigma,\sigma'_{-d^{-1}}},\lrangle{\varepsilon,\kappa'_{-d^{-1}}},0,b'),d')$,
    
    \hfill if $\sigma_\Dcal\in Q^4$ and $\sigma'\neq\sigma_d$.

    \comm{After releasing all its neighbors, some neighbors change their state.}
    
    \comm{Then update the information and leave the paused state.}
\end{enumerate}

\subsubsection{Proof sketch}\label{sssec:7.1.proof}
First, we give the representation function $R:Q'\rw Q$ as follows:\\
$\forall\psi=(\sigma,\sigma_\Dcal,\kappa_\Dcal,b,\nu)\in Q'$, $R(\psi)=\sigma$.

\paragraph{$\Gamma\flw_R\Gamma'$}
By the same argument as in Section~\ref{ssec:4.proof}, Lemma~\ref{lem:1}, we only need to prove that for any $\alpha',\beta'\in\Gamma'$ such that $\alpha'\rw^1\beta'$, we have $R^*(\alpha')\rw^1R^*(\beta')$. For such $\alpha',\beta'$, except for the reactions in item $3$ of protocol \ptc{d-sCRN\_$\gg$\_async-CA}, the first component of $\psi\in Q'$ won't change, so $R^*(\alpha')=R^*(\beta')$. For those two reactions, since the neighbors of $\psi=(\sigma,\sigma_\Dcal,\varepsilon_\Dcal,0,0)$ will not perform state transition after being locked, $\sigma_\Dcal$ is exactly the state around $\psi$ in configuration $\alpha'$. So $R^*(\sigma',\sigma_\Dcal,\varepsilon_\Dcal,0,1)=\sigma'=f(\sigma,\sigma_\Dcal)=f(R^*(\sigma,\sigma_\Dcal,\varepsilon_\Dcal,0,0),\sigma_\Dcal)$. This implies $R^*(\alpha')\rw^1R^*(\beta')$.

\paragraph{$\Gamma'\models_R\Gamma$}
For any $\alpha\in\Acal(\Gamma)$, let $\Pi=\{\alpha'\in\Acal(\Gamma'):R^*(\alpha')=\alpha\}$ be the set of all configurations that map to $\alpha$. Suppose that $\alpha$ can be produced by triggering a sequence of cells in $\Gamma'$, then we simulate it by having these cells lock their neighbors, perform state transition, and release all their neighbors one by one. It is obvious that the resulting configuration maps to $\alpha$. Then it suffices to show that for every $\beta\in\Acal(\Gamma)$ such that $\alpha\rw^1\beta$ and for every $\alpha'\in\Pi$, there exists a sequence of reactions in $\Gamma'$ such that $\alpha'\rw\beta'$, where $R^*(\alpha')=\alpha, R^*(\beta')=\beta$. Assume that $\beta$ can be produced from $\alpha$ by applying local function $f$ at a cell $v$. Since each lock between a pair of species can be unlocked unconditionally, we can always take the configuration back to the one with no lock. At this point, it is possible to perform state transition on $v$ by locking all its neighbors unless $v$ is a pause species where $f$ has no effect on it. Therefore, we conclude that $\alpha'\rw\beta'$ for some $R^*(\beta')=\beta$.

\paragraph{$\Gamma\Lrw_R\Gamma'$}
By the above explanation, it is obvious that $\{R^*(\alpha'):\alpha\in\Acal(\Gamma')\}=\Acal(\Gamma)$. Notice that there exists a fixed point $\alpha$ in $\Gamma$ that is reachable if and only if there exists a reachable configuration $\alpha'$ in $\Gamma'$ such that $R^*(\alpha')=\alpha$ and all species are paused by letting every cell perform a state transition according to the local function. As a result, $\{R^*(\alpha'):\alpha\in\Acal_*(\Gamma')\}=\Acal_*(\Gamma)$ as well.

\subsection{Simulate directed sCRN by non-deterministic async-CA}\label{sec:sCRN-CA}
\begin{theorem}\label{thm:CA-dsCRN}
Given a directed sCRN $\Gamma=(Q,S,R)$, there exists a non-deterministic, asynchronous CA $\Gamma'=(Q',\Ncal,f)$ which simulates $\Gamma$.
\end{theorem}
\subsubsection{Simulation overview}
Given a d-sCRN $\Gamma=(Q,S,R)$, we give our non-deterministic async-CA $\Gamma'$ the same initial configuration, which is the seed $s$ at a specific cell and all the other cells contain $\Ocal$. To simulate a sCRN, we have to construct the local function $f$. So we assign the outcome of $f$ on every possible neighborhood. For unimolecular reaction, we give the protocol in Section~\ref{sssec:7.2.ptc} \ptc{Uni\_Reactions}, by just having a cell change its state regardless of what its neighbors are. For bimolecular reactions, notice that in sCRN, two species change their state at the same time, but there's always one cell changing state at a time in cellular automata. Therefore we have to "tie" a pair of adjacent cells together. Let a cell $\sigma$ non-deterministically pick a neighbor $\sigma'$ in direction $d$ and "invite" it if there exists a reaction in $R$ that uses $\sigma,\sigma'$ as reactants. $\sigma$ becomes the \emph{invite state} $(\invsf,\sigma,\sigma',d)$. After receiving several invitations, $\sigma'$ non-deterministically choose one of them to "accept" and turns into an \emph{accept state} $(\accsf,\sigma',\varsigma,d')$, which means that it accepts the invitation from a cell with state $\varsigma$ in direction $d'$. Equivalently, other invitations are viewed as being "rejected". Then the "invite-accept pair" can change their state according to the reaction respectively. For simplicity we let a cell implicitly lock all of its neighbors as soon as it turns into an invite state or an accept state. The last step is to tell all the rejected cells to give up and return to their previous state. The protocol is given in Section~\ref{sssec:7.2.ptc} \ptc{Bi\_Reactions}. Figure~\ref{fig:bimolecular_reaction} shows a simple example, and the correctness proof is given in Section~\ref{sssec:7.2.proof}.

\begin{figure}[htbp]
    \centering
    \begin{subfigure}{\textwidth}
        \centering
        \includegraphics[width=0.4\textwidth]{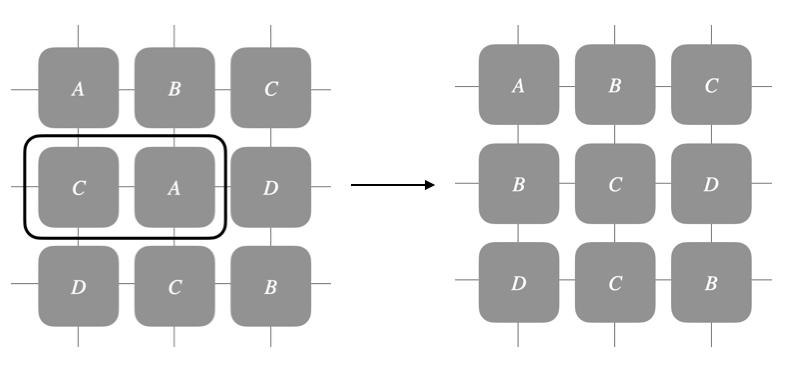}
        \caption{Bimolecular reaction in sCRN}
        \label{fig:sCRN}
    \end{subfigure}%
    
    \begin{subfigure}{\textwidth}
        \centering
        \includegraphics[width=\textwidth]{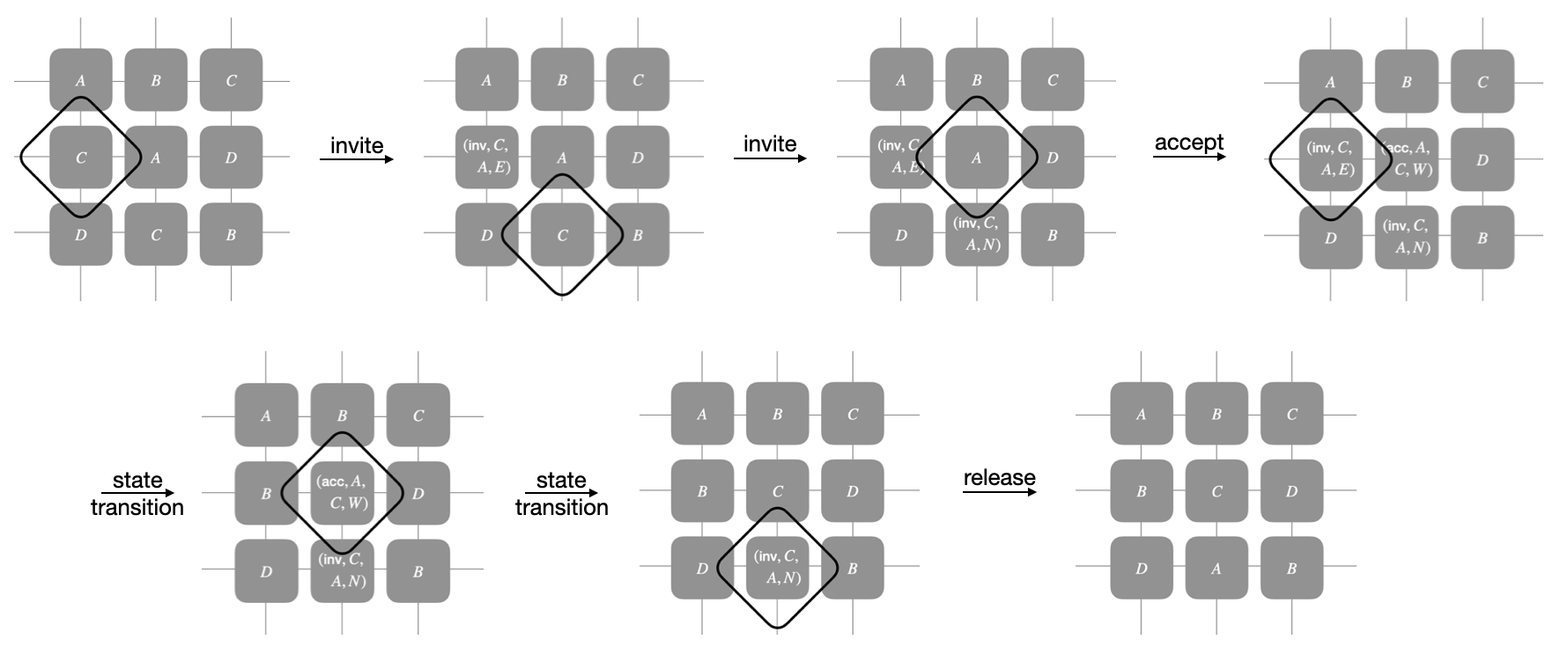}
        \caption{Simulation by async-CA}
        \label{fig:bimolecular_react}
    \end{subfigure}%
    \caption{Simulating sCRN reaction $C+A\rw B+C$ on non-deterministic async-CA}
    \label{fig:bimolecular_reaction}
\end{figure}
\subsubsection{Protocols}\label{sssec:7.2.ptc}
\paragraph{\ptc{Uni\_Reactions}}
This protocol simulates the unimolecular reactions in $\Gamma$. First, we let all the states in $Q$ be contained in $Q'$. For all $\psi,\psi'\in Q$ such that $(\psi,\psi',\odot)\in R$, we add $\sigma'$ to $f(\psi,\psi_N,\psi_E,\psi_S,\psi_W)$ for every $\psi_\Dcal\in Q^4$.

\paragraph{\ptc{Bi\_Reactions}}
To be more precise, for any species $\psi\in Q$ in sCRN, we define the set $nR(\psi,d)=\{\psi':\exists \psi_x,\psi_y \text{ s.t. }(\psi,\psi',\psi_x,\psi_y,d)\in R\}$ to be all possible species that can react with $v$ in direction $d$ according to $R$. A cell $v$ non-deterministically picks a direction $d$. If the neighbor $\vec{v'}=\vec{v}+\vec{v_d}$ contains a state representing a species $\psi'\in nR(\psi,d)$, then $v$ changes to an \emph{invite state} $(\invsf,\psi,\psi',d)$. On the other hand, $v'$ non-deterministically picks an invitation to accept, changes to an \emph{accept state} $(\accsf,\psi',\varphi,d')$ (if accepting $v$ then $\varphi=\psi,d'=d^{-1}$), and then the ''invite-accept'' pair of cells perform state transition respectively. A state involved in a bimolecular reaction can either invite others or accept the invitations it receives. So we discuss the following two cases:
\begin{enumerate}
    \item If a state $\psi\in Q$ has a neighbor $\psi_d\in nR(\psi,d)$, $\psi$ could turn into an invite state toward that neighbor. So we add the state $(\invsf,\psi,\psi_d,d)$ to $f(\psi,\psi_N,\psi_E,\psi_S,\psi_W)$ for all $\psi_{-d}\in Q^3$.
    \item If $\psi$ has been invited by some of its neighbors, it (non-deterministically) picks one to accept among all of these invitations and turns into the accept state. So for any neighbor $\psi_d=(\invsf,\psi',\psi,d^{-1})$, we add the state $(\accsf,\psi,\psi',d)$ to $f(\psi,\psi_N,\psi_E,\psi_S,\psi_W)$ for all $\psi_{-d}\in(Q\cup\{\text{invite states}\})^3$.
\end{enumerate}

As long as an invite-accept pair $((\invsf,\psi,\psi',d),(\accsf,\psi',\psi,d^{-1}))$ occurs and there exists $\psi_x,\psi_y\in Q$ s.t. $(\psi,\psi',\psi_x,\psi_y,d)\in R$, they change their state consecutively. So we add $\psi_x$ to $f((\invsf,\psi,\psi',d),\lrangle{(\accsf,\psi',\psi,d^{-1}),\psi_{-d}})$ for all $\psi_{-d}\in Q^3$, and add $\psi_y$ to $f((\accsf,\psi',\psi,d^{-1}),\lrangle{\psi_x,\psi'_{-d^{-1}}})$ for all $\psi'_{-d^{-1}}\in Q^3$.

It remains to deal with those invite states that have not been accepted. If an invite state $(\invsf,\psi,\psi',d)$ observes its neighbor in direction $d$ is not in state $\psi'$, it knows that its invitation has been rejected. So it turns back to its original state. There may be a situation where $\psi'$ has chosen another state to perform a bimolecular reaction but remained unchanged itself. In this case, we let $(\invsf,\psi,\psi',d)$ keep inviting $\psi'$. Therefore, we only need to add the state $\psi$ to $f((\invsf,\psi,\psi',d),\psi_N,\psi_E,\psi_S,\psi_W)$ for all $\psi_\Dcal\in Q^4$ and $\psi_d\neq\psi'$.

We summarize the above simulation in the following protocol.

\paragraph{\ptc{async-CA\_$\gg$\_d-sCRN}}
For all $\psi,\psi'\in Q$, we construct $f$ by the following protocol:
\begin{enumerate}
    \item For all $\psi_\Dcal\in Q^4$, add the following states to $f(\psi,\psi_N,\psi_E,\psi_S,\psi_W)$:
    \begin{itemize}
        \item[1.1] $\sigma'$, if $(\psi,\psi',\odot)\in R$. \comm{Unimolecular reactions.}
        \item[1.2] $(\invsf,\psi,\psi_d,d)$, if $\psi_d\in nR(\psi,d)$. \comm{Change to invite state.}
    \end{itemize}

    \item For all $\psi_\Dcal\in(Q\cup\{\text{invite states}\})^4$, add the state to $f(\psi,\psi_N,\psi_E,\psi_S,\psi_W)$:
    \begin{itemize}
        \item [2.1]$(\accsf,\psi,\psi',d)$, if $\psi_d=(\invsf,\psi',\psi,d^{-1})$. \comm{Change to accept state.}
    \end{itemize}

    \item For all $(\psi,\psi',\psi_x,\psi_y,d)\in R$:
    \begin{itemize}
        \item [3.1]Add $\psi_x$ to $f((\invsf,\psi,\psi',d),\lrangle{(\accsf,\psi',\psi,d^{-1}),\psi_{-d}})$ for all $\psi_{-d}\in Q^3$.
        \item [3.2]Add $\psi_y$ to $f((\accsf,\psi',\psi,d^{-1}),\lrangle{\psi_x,\psi'_{-d^-1}})$ for all $\psi'_{-d^{-1}}\in Q^3$.
    \end{itemize}
    \comm{Invite-accept pairs performing state transitions to simulate bimolecular reactions.}

    \item Add $\psi$ to $f((\invsf,\psi,\psi',d),\psi_N,\psi_E,\psi_S,\psi_W)$, for all $\psi_\Dcal\in Q^4$ and $\psi_d\neq\psi'$.
    
    \comm{Invitation has been rejected.}

\end{enumerate}

\begin{remark}
If we allow the representation function $R$ to map a neighborhood in the cellular automata system to a species in sCRN, then we define $R$ as the following (where $\sigma,\sigma',\sigma_x\in Q$):
\begin{itemize}
    \item $R((\invsf,\sigma,\sigma',d),\sigma_\Dcal)=\sigma$.
    \item $R((\accsf,\sigma,\sigma',d),\sigma_\Dcal)=\sigma$.
    \item $R(\sigma_x,\lrangle{(\accsf,\sigma',\sigma,d^{-1}),\sigma_{-d}})=\sigma'$. When the accepted state has not performed a state transition.
    \item $R(\sigma,\sigma_\Dcal)=\sigma$ for all the remaining cases.
\end{itemize}
By a similar explanation as above, we could see that $\Gamma\flw_R\Gamma'$, $\Gamma'\models_R\Gamma$, and $\Gamma\Lrw_R\Gamma'$ under this kind of representation function $R$.
\end{remark}

\subsubsection{Proof sketch}\label{sssec:7.2.proof}
First we give the representation function $R:Q'\rw Q$.\\
For all accept state $\sigma_\accsf$, $R(\sigma_\accsf)=\UNDsf$. For all invite state $\sigma_\invsf=(\invsf,\sigma,\sigma_d,d)$, $R(\sigma_\invsf)=\sigma$. Otherwise, $R(\sigma)=\sigma\ \forall\sigma\in Q$.

\paragraph{$\Gamma\flw_R\Gamma'$}
For any $\alpha',\beta'\in\Gamma'$ s.t. $\alpha'\rw\beta'$, if $R^*(\alpha'),R^*(\beta')\neq\UNDsf$, and $R^*(\alpha')\neq R^*(\beta')$, then $\alpha'\rw\beta'$ if and only if at least one of the following holds:
\begin{itemize}
    \item A unimolecular reaction is simulated directly by applying a local function on a cell (item $1.1$ in protocol \ptc{async-CA\_$\gg$\_d-sCRN}).
    \item Some invite-accept pairs is produced within the transformation from $\alpha'$ to $\beta'$, and all the accepted states have undergone their state transitions, meaning that every invite-accept pairs has completely simulate a bimolecular reaction in $\Gamma$ (item $3.2$ in protocol \ptc{async-CA\_$\gg$\_d-sCRN}).
\end{itemize}
Therefore, for $\alpha'\rw\beta'$ where $R^*(\alpha')\neq R*(\beta')$, $R^*(\alpha')\rw R*(\beta')$ corresponds to a sequence of reactions in the sCRN $\Gamma$.

\paragraph{$\Gamma'\models_R\Gamma$}
For any $\alpha\in\Acal(\Gamma)$, let $\Pi=\{\alpha\}$. $\alpha$ is achievable by simulating the unimolecular and bimolecular reactions by item $1.1$ and $3.2$ respectively. For all $\alpha'$ s.t. $R^*(\alpha')=\alpha$, there is no accept state in $\alpha'$, so we can trigger all the invite state in an arbitrary order from $\alpha$ since the invite state implicitly lock all of its neighbors. The result follows.

\paragraph{$\Gamma\Lrw_R\Gamma'$}
By the above explanation, it is obvious that $\{R^*(\alpha'):\alpha\in\Acal(\Gamma')\}=\Acal(\Gamma)\cup\{\UNDsf\}$. Notice that if there is no reactions that can be performed on a configuration $\alpha$, then there exists a configuration $\alpha'$ with no invite-accept pair and no rejected states s.t. $R^*(\alpha')=\alpha$ (which is equivalent to the termination of $\alpha'$) and vise versa. As a result, we have $\{R^*(\alpha'):\alpha\in\Acal_*(\Gamma')\}=\Acal_*(\Gamma)$.

%% file: sCRN-Amoe_v7.tex
\section{Simulation between amoebot and clockwise sCRN}\label{sec:amoe}
In this section, we show that clockwise sCRN and amoebot can simulate each other. The simulation for amoebot by clockwise sCRN is given in Section~\ref{ssec:Ameo-sCRN}, and the simulation for clockwise sCRN by amoebot is given in Section~\ref{ssec:sCRN-Ameo}

\subsection{Simulating amoebot by clockwise sCRN}\label{ssec:Ameo-sCRN}
\begin{theorem}\label{thm:csCRN-amoe}
Given an amoebot system $\Gamma=(Q,S,\delta)$, there exists a clockwise sCRN $\Gamma'=(Q',S',R')$ which simulates $\Gamma$.
\end{theorem}

\subsubsection{Simulation overview}\label{sssec:8.1.ov}
In the amoebot model, a particle performs state transitions and movements according to the flags placed by its neighbors, and it occupies two cells at a time if it is expanded. Therefore, the corresponding sCRN must have pairs of species that encode the same expanded particle, observe all the flags, and simulate all possible types of movement. Notice that an algorithm for the amoebot system may result in different, asymmetric configurations for clockwise and counterclockwise chirality, so it is necessary to 
use a sCRN with knowledge of the correct chirality (clockwise here). This could be accomplished by providing the surface with a predefined coloring at the beginning, or just using a clockwise sCRN. We now assume that $\Gamma'$ is given an initial pattern as in $\Gamma$. Recall that $Q=\Phi\times\Ocal\times D\times\Sigma^{10}$. The following are two steps to construct the simulation:
\begin{itemize}
    \item First, we construct a set of states $Q^*=\Phi\times\Ocal\times\Hcal\times\Sigma^{10}\times(\Ocal\cup\{\varepsilon\})$ that can be used to encode each particle in $\Gamma$.
    \item With $Q^*$, we can construct the corresponding reactions $R'$ and state set $Q'$ in $\Gamma'$ more easily.
\end{itemize}

We use a single species to represent a contracted particle, which is called the \emph{contracted species}. A pair of \emph{expanded species} consists of one \emph{tail species} and one \emph{head species}, and they together represent an expanded particle. Recall that expansion preserves local orientation; we observe that each expanded particle can be classified into $6$ \emph{types}, depending on which direction of cell a contracted particle expanded into when creating the expanded particle. More details are given in Section~\ref{sssec:8.1.P}.

To simulate a transition rule in the amoebot model, the corresponding species in sCRN must first observe all its neighbors' flags and then perform the corresponding movement.
For the observation part, we mainly follow the designing idea in the simulation of cellular automata, which is stated in Section~\ref{sec:CA-sCRN}. We'll briefly describe when to lock or unlock the neighbors and when to enter the pause state for each movement in Section~\ref{sssec:8.1.F}. For the second part, if the movement is $\idlesf$, it is just state transition and the simulation is straightforward. The protocol is given in \ptc{Idle}. For $\expandsf_i$, it can be simulated by having a contracted species react with a blank species $\Ocal$ in the $i$ direction, placing the flags in each direction carefully. The protocol is given in \ptc{Expansion}. We simulate contraction by having a pair of expanded species perform reactions with each other, after which one of them is turned into the blank species $\Ocal$, and the other becomes the contracted species. The contracting process is given in \ptc{Contraction}. For $\handoversf_i$, we simulate by slightly modifying the contraction and expansion protocols so that the contracted species can remember which cells it pushes in a pair of expanded species. After contracting out of that cell, the expanded species can tell the contracted species to expand into it. We define the \emph{switching cell} to be the one occupying by different particles before and after the handover operation. In the simulation, the contracted species first transforms the switching cell, say $v_s$, to a \emph{prepare species} and turns itself into a \emph{pushing species}. This makes the pair of expanded species contract out of $v_s$, leaving $v_s$ with a \emph{waiting species} which waits for the expansion of the locked contracted species. Details are provided in protocol \ptc{Handover}. These movements are simulated in Section~\ref{sssec:8.1.M}. The correctness proof is given in Section~\ref{sssec:8.1.proof}.

\subsubsection{Encoding particles}\label{sssec:8.1.P}
Here we explain how to represent a particle in $Q$ by species in $Q^*$. The labels $f_0,\cdots,f_9$ are directly copied from the original particle, so the contracted particles have $f_6=\cdots=f_9=\varepsilon$. $\Hcal=\{H,T,\varepsilon\}$ indicates whether head or tail the species is, where $\varepsilon$ means that it is a contracted species. Recall that in the amoebot system, the initial configuration consists of contracted particles, so each expanded particle is created by some contracted particle performing $\expandsf_i$, $i\in\Ocal$. The last variable indicate this $i$, and a contracted species always has a $\varepsilon$ there.

We map a contracted particle $\rho_C=(\phi,o,\varepsilon,f_0,\cdots,f_5,\varepsilon,\varepsilon,\varepsilon,\varepsilon)$ to a contracted species $\psi_C=(\phi,o,\varepsilon,f_0,\cdots,f_5,\varepsilon,\varepsilon,\varepsilon,\varepsilon,\varepsilon)\in Q^*$ naively. And each expanded particle $\rho_E=(\phi,o,d,f_0,\cdots,f_9)$ is mapped to a  pair of \emph{type-$i$ species}
\[(\psi_T,\psi_H)=\left((\phi,o,T,f_0,\cdots,f_9,i),(\phi,o,H,f_0,\cdots,f_9,i)\right)\text{ for some }i\in\Ocal.\]
We call $\psi_T$ the tail species, $\psi_H$ the head species, and we also call the corresponding expanded particle the \emph{type-$i$ particle}. Notice that once $i$ and $h\in\Hcal$ are given, the tail direction $d$ is determined as well. See Figure~\ref{fig:types} for an example of type-$0$ and type-$4$ species. The light gray circle represents the tail, and the dark gray circle represents the head. The one with an orange frame means that its local $0$ fits the local $0$ of the expanded particle. Such a mapping from $Q$ to $Q^*\cup (Q^*)^2$ is denoted as $F$, where $F(\rho_C)=\psi_C$, $F(\rho_E)=(\psi_T,\psi_E)$.

\begin{figure}[htbp]
    \centering
    \begin{subfigure}{0.2\textwidth}
        \centering
        \includegraphics[width=\textwidth]{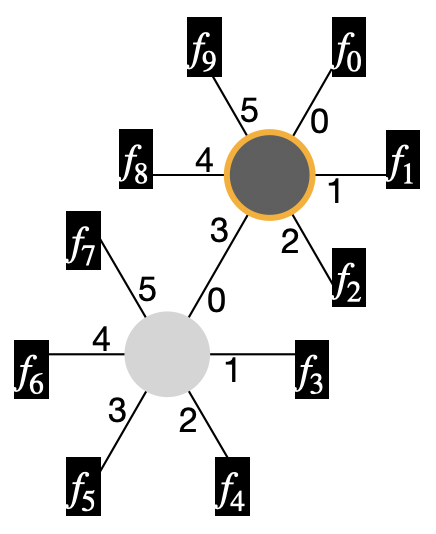}
        \caption{Type $0$.}
        \label{fig:type0}
    \end{subfigure}%
    \begin{subfigure}{0.25\textwidth}
        \centering
        \includegraphics[width=\textwidth]{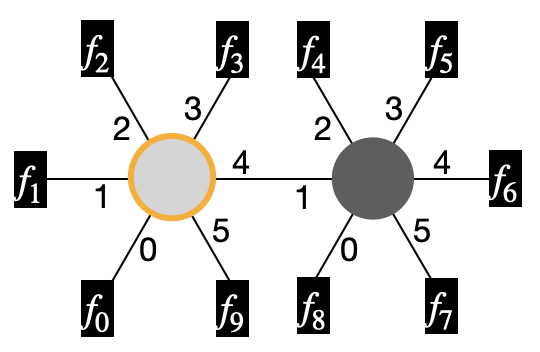}
        \caption{Type $4$.}
        \label{fig:type4}
    \end{subfigure}%
    \caption{Types of expanded species representing expanded particles}
    \label{fig:types}
\end{figure}

\subsubsection{Observing flags}\label{sssec:8.1.F}
Here we explain how to adapt the idea in Section~\ref{sec:CA-sCRN} to fit the purpose of observing all the flags around a particle. Similarly, a release species can only unlock its neighbors, and a paused species cannot lock any other species. We allow a species to lock others only if it has not been locked by anyone else, and most of the locks can be cancelled unconditionally except for some cases we'll specify in the following description.

For a contracted species $\psi_C$, it keeps locking and observing its neighbors.
\begin{itemize}
    \item If all its neighbors have been locked, $\psi_C$ can perform $\idlesf$. Note that if no state transition is needed to be performed at this point, $\psi_C$ turns itself into the paused release state.
    \item If the neighbor in direction $i$ is not locked by any other species, and the rest of its neighbors have all been locked by $\psi_C$, it can perform either an expansion or a handover toward direction $i\in\Ocal$.
    \begin{itemize}
        \item If an expansion is performed, the expanded species are produced in the release state.
        \item If a handover is performed, a pushing species and a prepare species appear first. A pushing species cannot unlock any neighbors, and the prepare species cannot lock the pushing species or its partner in the expanded species. Then the expanded species contract out of the prepare species as long as all the other neighbors have been locked. This results in a contracted species and a waiting species, where the contracted one is produced in the release state, and the waiting species cannot unlock any species. The last step is to let the pushing species and waiting species simulate the expansion, and both of them become a release state at the same time.
    \end{itemize}
\end{itemize}

For a pair of expanded species $\psi_E$, they keep locking and observing all their neighbors except for each other, once the neighbors around them are all locked, they can either idle or perform a contraction.
\begin{itemize}
    \item If the movement is $\idlesf$ and there is no state transition needed to be performed, then the species turn themselves into the paused release state.
    \item If the movement is a contraction, then both the contracted species and blank species resulting from the contraction are produced in the release state.
\end{itemize}

Following the design when simulating cellular automata, a paused species leaves the pause state if any of its neighbors changes, which includes turning into a pushing, prepare, or waiting species. The process described in this section is similar to the one in Section~\ref{sec:CA-sCRN}, so we separate it from the simulation of movements, and we are not giving the technical details or the sketch of proof of this part in this paper. Further, we assume from now on, every species knows all the flags facing it at any point.

\subsubsection{Movements}\label{sssec:8.1.M}
The goal now remains to simulate any legal movement in $\Gamma$. See Figure~\ref{fig:movements} for some examples of these operations.

\begin{figure}[htbp]
    \centering
    \includegraphics[width=\textwidth]{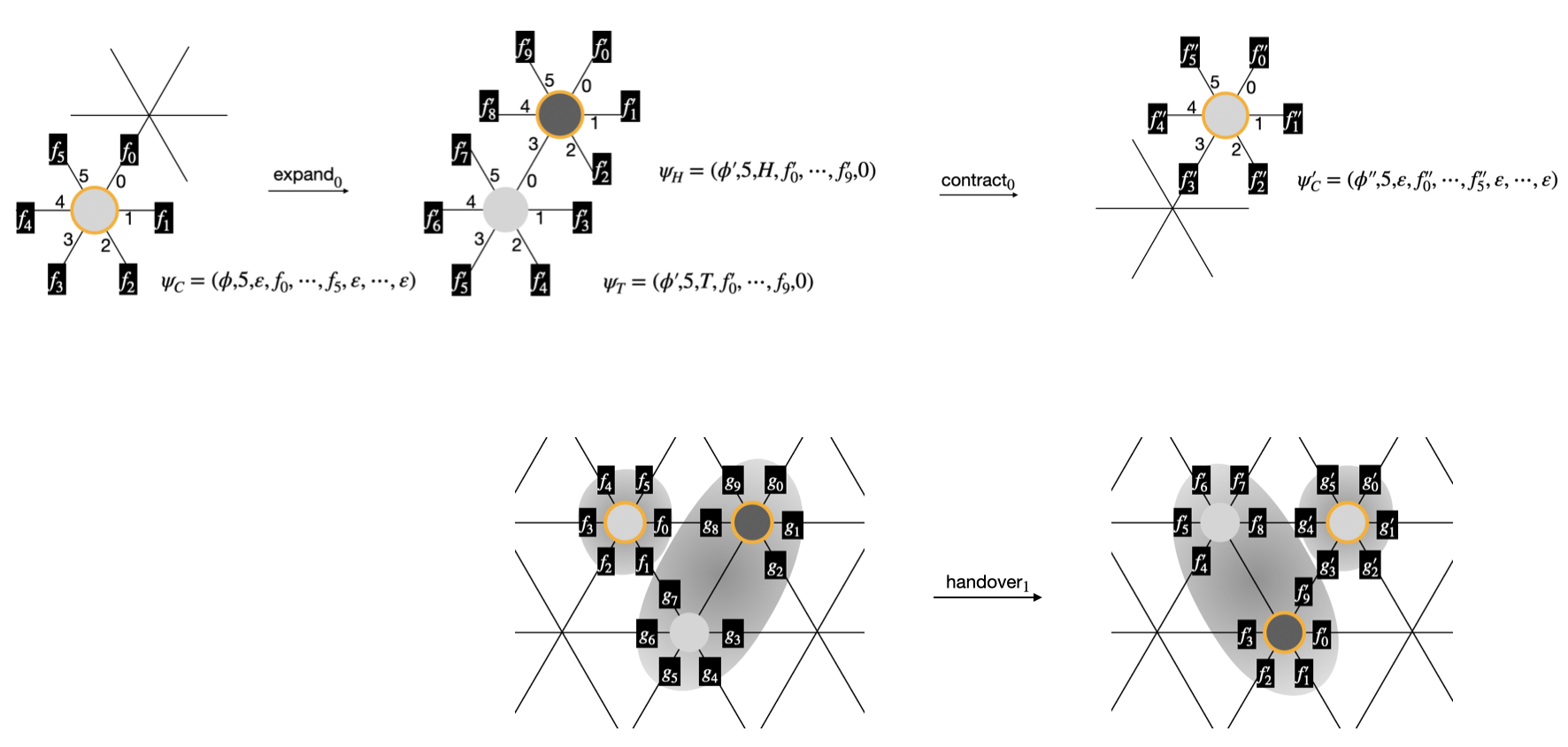}
    \caption{Examples of movements}
    \label{fig:movements}
\end{figure}

\paragraph{\ptc{Idle}}
It is just a state transition. Assume that a contracted particle $\rho_C$ becomes $\rho'_C$ by performing $\idlesf$, and that $\psi_C=F(\rho_C)$, $\psi'_C=F(\rho'_C)$, we add the unimolecular reaction $(\psi_C,\psi'_C,\odot)$ to $R'$. For an expanded particle $\rho_E$ that becomes $\rho'_E$, assume that $(\psi_T,\psi_H)=F(\rho_E)$ and $(\psi'_T,\psi'_H)=F(\rho'_E)$, we add the bimolecular reaction $(\psi_T,\psi_H,\psi'_T,\psi'_H,k)$ to $R'$ where $k$ is the type of $\rho_E$.

\paragraph{\ptc{Expansion}}
To simulate a contacted particle $\rho_C\in Q$ that expands toward direction $i$ and becomes an expanded particle $\rho_E\in Q$, assuming that $\psi_C=F(\rho_C)$ and $(\psi_T,\psi_H)=F(\rho_E)$, we add $(\psi_C,\Ocal,\psi_T,\psi_H,i)$ to $R'$.

\paragraph{\ptc{Contraction}}
To simulate a type-$k$ particle performing $\contractsf_i$, we first see whether it is contracting into the head or the tail.
\begin{itemize}
    \item When $(k,i)\in C_T=\{(0,5),(1,6),(2,9),(3,0),(4,1),(5,4)\}$, a type-$k$ particle performing $\contractsf_i$ is actually contracting into the tail.
    \item When $(k,i)\in C_H=\{(0,0),(1,1),(2,4),(3,5),(4,6),(5,9)\}$, a type-$k$ particle performing $\contractsf_i$ is actually contracting into the head.
\end{itemize}
Therefore, given an expanded particle $\rho_E\in Q$ which is contracting toward direction $i$ and is turning into a contracted particle $\rho_C\in Q$, we have the following reactions in $R'$ assuming $\psi_C=F(\rho_C)$ and $(\psi_T,\psi_H)=F(\rho_E)$:
\begin{flalign*}
    &(\psi_T,\psi_H,\psi_C,\Ocal,k), \text{ for all }(k,i)\in C_T.\\
    &(\psi_T,\psi_H,\Ocal,\psi_C,k), \text{ for all }(k,i)\in C_H.&&
\end{flalign*}

\paragraph{\ptc{Handover}}
A handover is equivalent to a contraction and an expansion happening simultaneously. In particular, we view $\handoversf_i$ as $(\expandsf_i,\contractsf_j)$ performing on a pair of adjacent particles. In the sCRN, it is impossible to change the states of three cells at the same time, so we use an additional variable $l\in\Ocal\cup\{\varepsilon,w\}$ to ensure that the corresponding contraction and expansion both happen eventually. If $l\in\Ocal$, this means a pushing species intending to push a tail or head species toward the $l$ direction, otherwise, $l=\varepsilon$ represents a prepare species being pushed by a pushing species. When the expanded species has contracted out of the switching cell, $l$ is set to be $w$ that stands for the waiting species.

Assume that it is a contracted particle $\rho_C$ pushing an expanded type-$k$ particle $\rho_E$ toward direction $i$, and they become $\rho'_E$ and $\rho'_C$ eventually. Let $F(\rho_C)=\psi_C$, $F(\rho'_C)=\psi'_C$, $F(\rho_E)=(\psi_T,\psi_H)$, and $F(\rho'_E)=(\psi'_T,\psi'_H)$. We describe the simulation protocol by the following three parts:
\begin{enumerate}
    \item We have $\rho_C$ remember $i$, the direction of the switching cell it is pushing. Let the species located in the switching cell be $\psi_h$ for some $h\in{H,T}$ ($h$ is fixed given $(k,j)$), then we have $\psi_C$ become a pushing species and $\psi_h$ become a prepare species $(\psi_h,\varepsilon)$. At this point, $(\psi_C,i)$ and $(\psi_h,\varepsilon)$ temporarily ``tie'' each other. We add the following reaction to $R'$:
    \begin{flalign*}
        &(\psi_C,\psi_H,(\psi_C,i),(\psi_H,\varepsilon)),\text{ for }(k,j)\in C_T.\\
        &(\psi_C,\psi_T,(\psi_C,i),(\psi_T,\varepsilon)),\text{ for }(k,j)\in C_H.&&
    \end{flalign*}
    At the same time there may be another contracted particle pushing $\rho_E$, but these two pushes cannot be performed simultaneously. To avoid this situation we also add the reverse reactions $(\psi_C,\psi_h,(\psi_C,i),(\psi_h,\varepsilon))^{-1}$ to $R'$.

    \item After the prepare species $\psi_h$ appears within the expanded species, the expanded species contracts out of $\psi_h$, and leaves a mark $W$ on $\psi_h$, which serves as a signal that tells $(\psi_C,i)$ to perform the expansion. Suppose that $\rho_E$ is a type-$k$ particle, we add the following reactions to $R'$:
    \begin{flalign*}
        &(\psi_T,(\psi_H,\varepsilon),\psi'_C,(\psi_H,w),k),\text{ for }(k,j)\in C_T.\\
        &((\psi_T,\varepsilon),\psi_H,(\psi_T,w),\psi'_C,k),\text{ for }(k,j)\in C_H.&&
    \end{flalign*}
    Note that $k$ is known given $\rho_E$, which, together with $j$ determines whether head or tail it is contracting into.

    \item After $(\psi_C,i)$ observes its neighbor in direction $i$ becoming $(\psi_h,w)$, it complete the expanding process. So we add the following reaction to $R'$:
    \begin{flalign*}
        ((\psi_C,i),(\psi_h,w),\psi'_T,\psi'_H,i).&&
    \end{flalign*}
\end{enumerate}

We summarize the simulation of all the movements in the following protocol.
\paragraph{\ptc{Simulate\_Movements\_by\_c-sCRN}}
For any transition define by $\delta$, we simulate it according to the movement being performed. For all contracted particles $\rho_C,\rho'_C\in Q$, expanded type-$k$ particles $\rho_E,\rho'_E\in Q$, a vector of flags $\vec{f''}=(f''_0,\cdots,f''_9)$ around the particle that is about to move, we assume that $F(\rho_C)=\psi_C$, $F(\rho'_C)=\psi'_C$, $F(\rho_E)=(\psi_T,\psi_H)$, and $F(\rho'_E)=(\psi'_T,\psi'_H)$. For simplicity, let
$\delta(\cdot)=\delta(\cdot,\vec{f''})$. Add the following reactions to $R'$:
\begin{enumerate}
    \item $(\psi_C,\psi'_C,\odot)$, if $(\rho'_C,\idlesf)=\delta(\rho_C)$.\\
    $(\psi_T,\psi_H,\psi'_T,\psi'_H,k)$, if $(\rho'_E,\idlesf)=\delta(\rho_E)$.\comm{Idle.}
    
    \item $(\psi_C,\Ocal,\psi_T,\psi_H,i)$, if $(\rho_E,\expandsf_i)=\delta(\rho_C)$ for some $i\in\{0,\cdots,5\}$.\comm{Expansion.}
    
    \item $(\psi_T,\psi_H,\psi_C,\Ocal,k)$, if $(\rho_C,\contractsf_i)=\delta(\rho_E)$ and $(k,i)\in C_T$.\comm{Contract into tail.}\\
    $(\psi_T,\psi_H,\Ocal,\psi_C,k)$, if $(\rho_C,\contractsf_i)=\delta(\rho_E)$ and $(k,i)\in C_H$.\comm{Contract into head.}

    \item If $(\hat{\rho}_C,\handoversf_i)=\delta(\rho_C)$ s.t. $\hat{\rho}_C$ puts the information on its flags that makes $\rho_C$ expand along direction $i$ and become $(\rho'_E)$, and makes $\rho_E$ contract toward direction $j$ and change to $\rho'_C$, then we add the following reactions to $R'$:
    
    $(\psi_C,\psi_H,(\psi_C,i),(\psi_H,\varepsilon))$, $(\psi_C,\psi_H,(\psi_C,i),(\psi_H,\varepsilon))^{-1}$, for $(k,j)\in C_T$.\\
    $(\psi_C,\psi_T,(\psi_C,i),(\psi_T,\varepsilon))$, $(\psi_C,\psi_T,(\psi_C,i),(\psi_T,\varepsilon))^{-1}$, for $(k,j)\in C_H$.\\
    $(\psi_T,(\psi_H,\varepsilon),\psi'_C,(\psi_H,w),k)$, for $(k,j)\in C_T$.\\
    $((\psi_T,\varepsilon),\psi_H,(\psi_T,w),\psi'_C,k)$, for $(k,j)\in C_H$.\\
    $((\psi_C,i),(\psi_h,w),\psi'_T,\psi'_H,i)$, for $h\in\{T,H\}$.\comm{Handover}
\end{enumerate}

\subsubsection{Proof sketch}\label{sssec:8.1.proof}
First we define the representation function. For any contracted species, it is mapped to the contracted particle it represents, so $R(\psi_C)=F^{-1}(\psi_C)$ for all $\psi_C\in\{\text{contracted species}\}$. For tail and head species in a type-$i$ T-H pair s.t. $F(\rho_E)=(\psi_T,\psi_H)$ for some expanded particle $\rho_E$, let $R(\psi_T)=R(\psi_H)=\rho_E$. This is well-defined since we can decode $\rho_E$ by any one of the $\psi_T$ and $\psi_H$. For the remaining species, let $R(\Ocal)=\nullsf$, and $R((\psi,l))=R((\psi,\varepsilon))=R((\psi,w))=\UNDsf\ \forall l\in\Ocal$. The contraction, expansion, or handover, toward each direction, corresponds one-to-one to a protocol in Section~\ref{sssec:8.1.M}.

\paragraph{$\Gamma\flw_R\Gamma'$}
For any $\alpha'\rw\beta'$, it is a sequence consists of contraction, expansion, or a reaction in the sub-protocol simulating handover. If it was the last case, the configuration produced by one step of reaction is not mapped to $\UNDsf$ only if it performs item $3$ in protocol \ptc{Handover}, which means that a handover has complete. Therefore, if $R^*(\alpha'),R^*(\beta')\neq\UNDsf$, then $R^*(\alpha')\rw R^*(\beta')$ since it can be achieved by a sequence of movements.

\paragraph{$\Gamma'\models_R\Gamma$}
For any $\alpha\in\Acal(\Gamma)$, let $\Pi=\{\alpha\}$. $\alpha$ is achievable by simulating the movements one by one according to each sub-protocol. Notice that for any $\alpha''$ s.t. $R^*(\alpha'')=\alpha$, there cannot be any prepare species or waiting species in $\alpha''$. So $\alpha''$ must be a configuration s.t. all sequence of reactions that simulate a single movement is completely performed. Therefore $\alpha\rw\alpha''$.

\paragraph{$\Gamma\Lrw_R\Gamma'$}
By the above explanation, it is obvious that $\{R^*(\alpha'):\alpha\in\Acal(\Gamma')\}=\Acal(\Gamma)\cup\{\UNDsf\}$. Notice that the ability of performing a movement in $\Gamma$ is equivalent to the ability of performing the sub-protocol corresponds to that movement. (Although in the handover operation there might be more than one contracted species pushing the same expanded species, item $2$ of protocol \ptc{Handover} are reversible reactions, so it will be only one species ''accepted'' by the expanded species eventually.) Therefore we have $\{R^*(\alpha'):\alpha\in\Acal_*(\Gamma')\}=\Acal_*(\Gamma)$.

\subsection{Simulate clockwise sCRN by amoebot}\label{ssec:sCRN-Ameo}
\begin{theorem}\label{thm:amoe_csCRN}
Given a clockwise sCRN $\Gamma=(Q,S,R)$ , there exists an amoebot system $\Gamma'=(Q',S',\delta)$ which simulate $\Gamma$.
\end{theorem}

\subsubsection{Simulation overview}
We use only contracted particles, and with the property that each particle can see every neighbors and perform state transition by the information in its neighborhood, we can use the same technique as the simulation of d-sCRN by non-deterministic async-CA (Section~\ref{sec:sCRN-CA}). Notice that the underlying lattice does not matter and we don't provide the precies protocol here.